\documentclass[aps,reprint,showpacs,onecolumn,tightenlines]{revtex4-2}
\usepackage{ulem}
 
\usepackage{appendix} 
\usepackage{epstopdf}

\usepackage{float}

\usepackage{amsmath, amsthm, amssymb, bm, bbm}
\usepackage{mathrsfs}
\usepackage{extarrows}
\usepackage{pifont}
\usepackage{diagbox}
\usepackage{xcolor}
\usepackage{multirow}
\usepackage{graphicx}
\usepackage{hyperref}
\hypersetup{colorlinks=true, linkcolor=blue, citecolor=blue, urlcolor=black}
\allowdisplaybreaks
\usepackage{adjustbox} 
\usepackage{braket}
\usepackage{pgfplots}
\usepackage{qcircuit}
\usepackage[noend,noline,ruled,linesnumbered]{algorithm2e}

\newtheorem{lemma}{\indent Lemma}
\newtheorem{theorem}{\indent Theorem}

\newtheorem{proposition}{Proposition}

\newcommand{\Ebb}{\mathbb{E}}
\newcommand{\Fbb}{\mathbb{F}}
\newcommand{\Ibb}{\mathbb{I}}

\newcommand{\Rbb}{\mathbb{R}}

\newcommand{\Ccal}{\mathcal{C}}

\newcommand{\Ecal}{\mathcal{E}}
\newcommand{\Fcal}{\mathcal{F}}

\newcommand{\Hcal}{\mathcal{H}}
\newcommand{\Ical}{\mathcal{I}}

\newcommand{\Lcal}{\mathcal{L}}

\newcommand{\Ocal}{\mathcal{O}}
\newcommand{\Pcal}{\mathcal{P}}

\newcommand{\Ucal}{\mathcal{U}}
\newcommand{\Vcal}{\mathcal{V}}
\newcommand{\Wcal}{\mathcal{W}}

\newcommand{\Zcal}{\mathcal{Z}}

\newcommand{\Wmat}{\bm{\mathrm W}}
\newcommand{\Xmat}{\bm{\mathrm X}}

\newcommand{\Zmat}{\bm{\mathrm Z}}

\newcommand{\var}{\mathrm{Var}}
\newcommand{\Cl}{\mathrm{Cl}}
\newcommand{\Clgroup}{\mathbf{Cl}}
\newcommand{\tr}[1]{\textnormal{Tr}\left( #1 \right)}

\newcommand{\abs}[1]{\left| #1 \right|}

\newcommand{\sbra}[1]{\left[ #1 \right]}
\newcommand{\pbra}[1]{\left( #1 \right)}
\newcommand{\cbra}[1]{\left\{ #1 \right\}}

\newcommand{\superket}[1]{| #1 \rangle\!\rangle}
\newcommand{\superbra}[1]{\langle\!\langle #1 |}
\newcommand{\superbraket}[1]{\langle\!\langle #1 \rangle\!\rangle}

\begin{document}
\title{State Similarity in Modular Superconducting Quantum Processors with Classical Communications}

\author{Bujiao Wu$^1$}
\thanks{These authors contributed equally to this work.}

\author{Changrong Xie$^{2, 1}$}
\thanks{These authors contributed equally to this work.}

\author{Peng Mi$^3$}
\thanks{These authors contributed equally to this work.}

\author{Zhiyi Wu$^{4, 1}$}
\thanks{These authors contributed equally to this work.}

\author{Zechen Guo$^{2, 1}$}
\author{Peisheng Huang$^{5, 1}$}
\author{Wenhui Huang$^{2, 1}$}
\author{Xuandong Sun$^{2, 1}$}
\author{Jiawei Zhang$^{2, 1}$}
\author{Libo Zhang$^{2, 1}$}
\author{Jiawei Qiu$^1$}
\author{Xiayu Linpeng$^1$}
\author{Ziyu Tao$^1$}
\author{Ji Chu$^1$}
\author{Ji Jiang$^{2, 1}$}
\author{Song Liu$^{1, 6}$}
\author{Jingjing Niu$^{1, 6}$}
\author{Yuxuan Zhou$^1$}

\author{Yuxuan Du$^7$}
\email{yuxuan.du@ntu.edu.sg}

\author{Wenhui Ren$^1$}
\email{wenhuiren@zju.edu.cn}

\author{Youpeng Zhong$^{1, 6}$}

\author{Tongliang Liu$^3$}
\email{tongliang.liu@sydney.edu.au}

\author{Dapeng Yu$^{1, 6}$}

\affiliation{
$^1$International Quantum Academy, Shenzhen 518048, China\\
$^2$Guangdong Provincial Key Laboratory of Quantum Science and Engineering, Southern University of Science and Technology, Shenzhen, Guangdong, China\\
$^3$School of Computer Science, University of Sydney, Australia\\
$^4$School of Physics, Peking University, Beijing 100871, China\\
$^5$School of Physics and Electronic-Electrical Engineering, Ningxia University, Yinchuan 750021, China\\
$^6$Shenzhen Branch, Hefei National Laboratory, Shenzhen 518048, China\\
$^7$College of Computing and Data Science, Nanyang Technological University, 639798, Singapore
}

\begin{abstract}
	As quantum devices continue to scale, distributed quantum computing emerges as a promising strategy for executing large-scale tasks across modular quantum processors. A central challenge in this paradigm is verifying the correctness of computational outcomes when subcircuits are executed independently following circuit cutting. Here we propose a cross-platform fidelity estimation algorithm tailored for modular architectures. Our method achieves substantial reductions in sample complexity compared to previous approaches designed for single-processor systems. We experimentally implement the protocol on modular superconducting quantum processors with up to 6 qubits to verify the similarity of two 11-qubit GHZ states. Beyond verification, we show that our algorithm enables a federated quantum kernel method that preserves data privacy. As a proof of concept, we apply it to a 5-qubit quantum phase learning task using six 3-qubit modules, successfully extracting phase information with just eight training samples. These results establish a practical path for scalable verification and trustworthy quantum machine learning of modular quantum processors.
\end{abstract}

\maketitle

\section{INTRODUCTION} 
Quantum computing is steadily progressing from proof-of-concept demonstrations to practical utilities across a range of disciplines such as chemistry and fundamental science~\cite{Arute2019Quantum, Zhong2020Quantum, Madsen2022Quantum, Zhu2022Quantum, bluvstein2022quantum, Cao2023GenerationOG, Kim2023Evidence,Guo2024BoostedFG}. Along this path,  neutral-atom and superconducting quantum processors with over a hundred qubits have been fabricated~\cite{bluvstein2024logical,GoogleQuantumAI2024QEC}, with experimental results confirming their capabilities in sampling tasks beyond the reach of classical computers~\cite{morvan2024phase,gao2024establishing}. Despite these advances, scalability has become an increasingly pressing challenge, as expanding a single massive quantum processor to millions of qubits while preserving precise control and high-fidelity connectivity presents formidable technical hurdles~\cite{bravyi2022future,sunami2025scalable}. A leading solution is captured by the principle ``scaling up quantum computing by scaling out", which exploits modular and distributed architectures to achieve scalable and fault-tolerant quantum computing faster and more feasibly.

 A promising modular approach is distributed quantum computing~\cite{cirac1999distributed,wehner2018quantum,caleffi2024distributed}, which relies on quantum links to coordinate computations across remote quantum processors. However, achieving entangled interconnections remains a long-term challenge, with current experimental demonstrations still limited to small-scale systems~\cite{niu2023low,Main2024Distributed,aghaee2025scaling}. To maximize the capabilities of near-term quantum processors, a more practical alternative is circuit cutting~\cite{mitarai2021constructing,piveteau2023circuit,Lowe2023fast}. Unlike distributed quantum computing, the coordination among modular quantum processors is completed by specific local quantum operations and classical post-processing~\cite{CarreraVazquez2024Combing}. Thanks to its flexibility, diverse quantum algorithms can be adapted to this formalism~\cite{gentinetta2024overhead}. Moreover, the feasibility of circuit cutting has been recently demonstrated, with two 127-qubit quantum processors used to generate a 134-qubit graph state~\cite{CarreraVazquez2024Combing}.

\begin{figure*}
    \centering
    \includegraphics[width=1.0\linewidth]{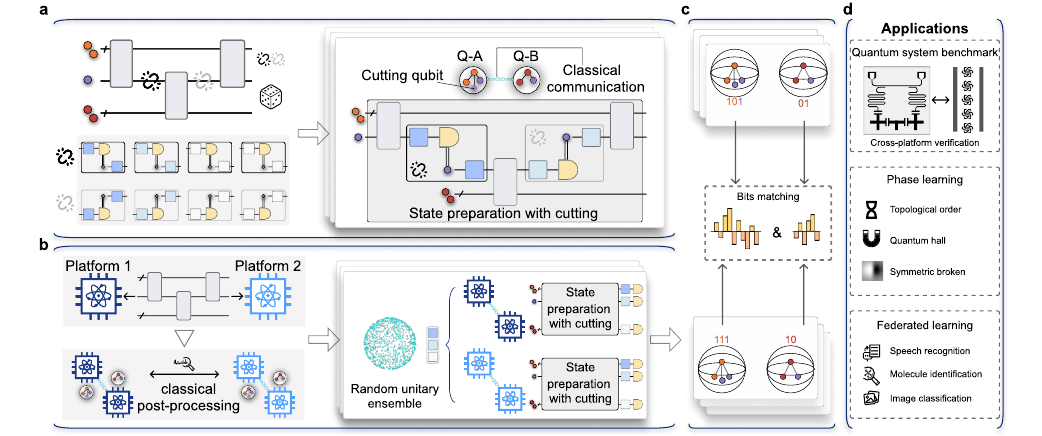}
     \caption{An overview of state similarity estimation between two 5-qubit states using circuit cutting across 4-qubit platforms (or 3-qubit platforms with qubit reuse) 
     \textbf{a.} Circuit cutting process and the first source of randomness in the subcircuits $\mathsf{Q}$-$\mathsf{A}$ and  $\mathsf{Q}$-$\mathsf{B}$: random Clifford circuits (indicated by boxes preceding yellow-colored measurements), applied at the cut locations. In this example, the quantum state with $n=5$ qubit is partitioned into $r=2$ subcricuits ($\mathsf{Q}$-$\mathsf{A}$ and  $\mathsf{Q}$-$\mathsf{B}$) with $k_2=2$ wire cuttings. Each cut involves $k_1=1$ qubit. 
\textbf{b.} Cross-platform fidelity estimation with circuit cutting is performed by preparing the quantum state using circuit cutting across modular $\mathsf{Q}$-$\mathsf{A}$ and  $\mathsf{Q}$-$\mathsf{B}$ devices. The second source of randomness, i.e., random unitaries sampled from a unitary ensemble, is applied before the final measurements.
\textbf{c.} Fidelity is estimated by classically post-processing the measurement outcomes, filtered by the sampled Bernoulli variables and corresponding Clifford circuits. \textbf{d.} Applications for the state similarity estimation algorithm with classical communication.}
    \label{fig:main_device}
\end{figure*}

Despite the progress, a fundamental challenge in using modular quantum processors with circuit cutting is how to rigorously certify the correctness of computations. The unwilling noise in modern quantum processors induces errors throughout the computation, making the final output unreliable and underscoring the need for verification procedures~\cite{eisert2020quantum}.  In addition, conventional verification techniques such as randomized benchmarking~\cite{Xue19Benchmarking,sung2021realization,bao2022fluxonium} or direct fidelity estimation~\cite{flammia2011direct,da2011practical} cannot be directly applied, as circuit cutting involves extra quantum operations, measurements, and classical post-processing.  This raises a key question: \textit{Can the outputs of circuit cutting be certified through provable and scalable methods}?

To address this, we propose an algorithm for estimating quantum state similarity, extending cross-platform verification to modular quantum processors. Conceptually, as shown in Fig.~\ref{fig:main_device}\textbf{b}, cross-platform verification assesses whether quantum states prepared on two isolated quantum platforms are consistent, without relying on trust in either device~\cite{Elben2020cross}. While the task appears to be a natural generalization, our theoretical analysis reveals a separation in sample complexity between our proposal and prior methods developed for single-processor architectures. Specifically, in contrast to prior methods that require at least the order of $\sqrt{2^n}$ samples with $n$ being the qubit number~\cite{Anshu2022Distributed}, our proposal achieves a slightly exponential scaling
with $n$ in a class of practically relevant scenarios. This separation suggests a promising path toward scalable cross-platform verification for large-scale modular quantum processors. We further exhibit the effectiveness of our approach on a superconducting quantum processor. That is, by partitioning the available qubits into distinct parts, the device is configured to emulate modular quantum processors, enabling cross-platform verification under the circuit-cutting setting. In this way, the cross-platform fidelity of Greenberger–Horne–Zeilinger (GHZ) states in the circuit cutting setting up to $11$ qubits is evaluated using modular quantum processors up to $6$ qubits.

As a secondary result, we extend our approach to support trustworthy quantum machine learning~\cite{du2025quantum}, with potential applications in domains where data confidentiality is critical, such as medicine and finance~\cite{yang2019federated}. In particular, our method enables the implementation of federated quantum kernels such that the kernel matrix is constructed in a privacy-preserving manner. Experimental results on a 5-qubit phase learning problem exhibit that the proposed federated quantum kernel can accurately predict quantum phase transitions. The achieved results provide valuable insights into how circuit cutting can be used to advance quantum machine learning.

\section{MAIN RESULTS}

\noindent The crux of enabling large-scale quantum computation in modular architectures lies in the circuit cutting technique~\cite{piveteau2023circuit,mitarai2021constructing,Lowe2023fast,CarreraVazquez2024Combing}. As illustrated in Fig.~\ref{fig:main_device}, a large-scale quantum circuit corresponding to a given program is partitioned into smaller segments. This divide-and-conquer strategy substantially reduces the number of qubits required on each quantum module, allowing existing quantum processors to be integrated for solving more challenging tasks. Depending on the cutting strategy, current circuit cutting methods can be primarily categorized into two classes: gate cutting~\cite{piveteau2023circuit,mitarai2021constructing} and wire cutting~\cite{Lowe2023fast}. In gate cutting, entangling operations between subcircuits are replaced by probabilistic or classically coordinated procedures. In wire cutting, the circuit is partitioned by measuring and reinitializing qubits along specific wires. Our focus in this work is on wire-cutting techniques, thanks to their more favorable scaling properties in large-scale modular quantum processors. 

Let us briefly introduce the mechanism of wire cutting (see Supplementary Information (SI) I for details). As shown in Fig.~\ref{fig:main_device}\textbf{a}, wire cutting decomposes a given quantum circuit into smaller and independently executable subcircuits by partitioning qubit wires at designated locations.  Suppose the input $n$-qubit state is partitioned into $r$ subcircuits using $k_2$ wire cuttings, where each cut involves $k_1$ qubit wires. 
For the $l$-th cut with $l\in [k_2]$, the identity channel on the  $k_1$ cut qubits between two involved subcircuits is replaced using equivalent measure-and-prepare channels, i.e., 
\begin{equation}\label{eqn:wire-cut}
\mathcal{I}_{2^{k_1}\times 2^{k_1}}(X)=\sum_{i}\tr{M_i X}\kappa_i, 
\end{equation}
 with $\{M_i\}$ being informationally complete measurements and $\{\kappa_i\}$ being the corresponding reinitialized states. An effective way to implement $\{M_i\}$ is to apply random unitaries sampled from the Clifford group, followed by the computational basis measurement. More precisely, the relevant operation is denoted by $\left( \bm{z}^{(l)}, U_{\bm{z}^{(l)}}, p_{\bm{z}^{(l)}} \right)$, indicating that when the Bernoulli variable associated with this cutting equals $\bm{z}^{(l)}\in\{0, 1\}^{k_1}$, the Clifford unitary $U_{\bm{z}^{(l)}}$ applied to the cut qubits is with probability $p_{\bm{z}^{(l)}}$, and the measurement outcomes are $\bm{c}^{(l)}$. The  variable $\bm z^{(l)}$ and the obtained bit-string $\bm{c}^{(l)}$ indicate the initial state $\kappa_i$ for the next sub-circuit. The target output state is reconstructed by repeatedly measuring the qubit state in one subcircuit and reinitializing it in the next.

 A potential solution for certifying the reliability of computations performed by modular quantum processors with wire cutting is using cross-platform verification protocols~\cite{Elben2020cross}. Conceptually, it assesses the similarity between quantum states prepared independently on different platforms executing the same quantum program. Denote the output $n$-qubit states from two quantum platforms as $\rho$ and $\sigma$, the cross-platform fidelity yields
\begin{equation}\label{eqn:cross-fide}
    \Fcal\pbra{\rho, \sigma} = \frac{\tr{\rho\sigma}}{\sqrt{\tr{\rho^2}\tr{\sigma^2}}}.
\end{equation}
For single-processor settings, recent theoretical results show that for arbitrary quantum states, the sample complexity of estimating $\Fcal$ scales at least sub-exponentially with $n$~\cite{Anshu2022Distributed}, but becomes efficient when the states exhibit certain structure, such as low magic~\cite{hinsche2024efficient}. However, prior algorithms and the corresponding results do not directly extend to modular architectures, due to the additional local quantum operations and classical communication involved. How to perform cross-platform verification for modular quantum processors remains unexplored.

  \begin{figure*}
\includegraphics[width=1.0\textwidth]{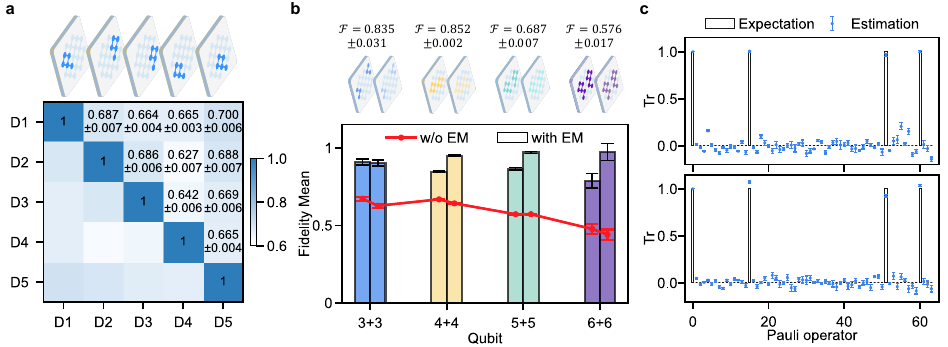}
\caption{Illustration for state similarity in modular superconducting quantum processors with classical communications. \textbf{a.} Cross-platform fidelities of 9-qubit GHZ states prepared on five distinct modular configurations, shown alongside the corresponding qubit layout. \textbf{b.} Cross-platform fidelities for GHZ states of up to 11 qubits distributed across two modules, together with fidelities relative to ideal GHZ states with/without error mitigation. \textbf{c.} A subset of tomography results for 5-qubit GHZ states prepared on Platforms 1 and 2, with error mitigation applied. The $x$-axis denotes a sequence of 64 selected Pauli operators $\cbra{P_i}$, and the $y$-axis shows the corresponding estimated values $\cbra{\tr{\rho P_i}}$ and their expected values.}
\label{fig:experiment_fid_main}
\end{figure*}

\subsection{Cross-platform verification in circuit-cutting case}
To address the above knowledge gap, we propose a cross-platform verification protocol for modular quantum processors with wire cutting. Given a specified $n$-qubit circuit,  the protocol proceeds in two key steps: (1) perform wire cutting on the circuit and independently execute the resulting subcircuits on Platform 1 and  Platform 2; and (2) apply classical post-processing to the measurement outcomes collected from each platform, as illustrated in Fig.~\ref{fig:main_device}\textbf{b,c}. We briefly present the implementation of each step subsequently, and defer more details and further improvements to SI~B and C. 

In the first step, both Platform 1 and  Platform 2 involve two sources of randomness. The first arises from the origin wire cutting procedure in Eq.~(\ref{eqn:wire-cut}). For this source, the Clifford unitaries sampled across different cuts and different platforms are independent. For notational clarity, we denote the sampled Bernoulli variables for all $k_2$ cuts as  $\bm{z}=\cup_{l=1}^{k_2} \bm{z}^{(l)}$ for Platform 1 and $\bm{z}'=\cup_{l=1}^{k_2} {\bm{z}'}^{(l)}$ for Platform 2. Each value of the Bernoulli variable is sampled $L$ times and the obtained bit-strings are denoted by $\bm{c}\in \{0,1\}^{k_1 \times k_2 \times L}$ and $\bm{c}'\in \{0,1\}^{k_1 \times k_2 \times L}$  for Platform 1 and Platform 2, respectively.

The second source of randomness originates from the application of random measurements to the output states of each subcircuit. Specifically, for the $j$-th subcircuit $\forall j \in [r]$, $N$ random unitaries $\{\mathcal{W}^{(i)}_j\}_{i=1}^N$ are sampled from the tensor product of a local 4-design~\cite{mele2024introduction}. In this setting, the sequence of sampled unitaries must be \textit{identical} across platforms to ensure consistent post-processing. We denote the corresponding bitstring outcomes for the $j$-th subcircuit, obtained by repeatedly applying $N$ random unitaries followed by computational measurements, as ${\bm s}_j$ and  ${\bm s}'_j$ for Platform 1 and Platform 2, respectively. 

In the second step, the proposed algorithm applies a classical post-processing to the collected $\bm{z}$, $\bm{c}$, $\{\Wcal_j^{(i)}\}$, and $\{{\bm s}_j\}$ to estimate $\tr{\rho \sigma}$. The estimator yields 
\begin{equation}
    \hat{v} := \sum_{\bm z, \bm z'}\frac{ f_{\bm 
 z,\bm  z'}}{N^r L^2} \sum_{\bm c,\bm c'} \prod_{j} 
\sum_{\bm s_j,\bm s'_j,i}  g_{\bm s_j,\bm s'_j}\hat{p}_{\bm s_j}^{\bm z, \bm c,\Wcal_j^{(i)}}\hat{q}_{\bm s'_j}^{\bm z',\bm c',\Wcal_j^{(i)}},
\label{eq:main_estimator}
\end{equation}
where $f_{\bm 
 z,\bm  z'}$ refers to a weight factor associated with $\bm z,k_1$, and $k_2$, 
 and $g_{\bm s_j,\bm s'_j}$ calculates the similarity between outcome $\bm s_j$ and $\bm s'_j$. Moreover, the quantity $\hat{p}_{{\bm s}_j}^{\bm{z}, \bm{c}, \Wcal_j^{(i)}}$ represents the estimated probability of observing outcome ${\bm s}_j$, conditioned on the values of $\bm{z}$, $\bm{c}$ and $\Wcal_j^{(i)}$. Refer to SI~C for the implementation details.

The theorem below quantifies how the estimation error of the proposed estimator relative to the true cross-platform fidelity scales with the sample complexity.
\begin{theorem}[Informal]\label{thm:1}
Following notations in Eqs.~\eqref{eqn:cross-fide} and~\eqref{eq:main_estimator}, when the random unitaries in the second randomness sources are sampled from a group yielding the tensor product of a unitary local $2$-design, $\hat{v}$ is an unbiased estimator of $\tr{\rho \sigma}$ with $\mathbb{E}[\hat{v}]=\tr{\rho \sigma}$. Moreover, when this group yields the tensor product of a unitary local $4$-design, and the number of measurements satisfies $\Ocal(( 2^{k_1+1} + 1)^{2k_2}6^{n/(2r)}/\varepsilon)$, we have  
\begin{equation}
 |\hat{v} - \tr{\rho \sigma}|\leq \varepsilon,
\end{equation}
with high success probability.
\label{thm:main_res} 
\end{theorem}

The achieved results convey two crucial implications. For the case $r = 2$, and $k_1$ and $k_2$ being two constants, the required number of copies of quantum states scales as $\Ocal\pbra{6^{n/4}/\varepsilon} \approx \Ocal\pbra{1.565^n/\varepsilon}$, which is significantly lower than the $\Ocal\pbra{2^n}$ scaling observed in Ref.~\cite{Elben2020cross} when $n$ continuously grows. Although Anshu et al.~\cite{Anshu2022Distributed} demonstrated that $\Ocal\pbra{\sqrt{2^n}/\varepsilon}$ copies suffice, their protocol relies on applying a global random unitary, which poses significant implementation challenges for near-term quantum devices. Furthermore, when $r \geq 3$, and $k_1$ and $k_2$ are constants, our algorithm achieves a sub-exponential improvement over all existing cross-platform algorithms, while eliminating the need for global random unitaries.

We introduce two modifications to further enhance the effectiveness of the proposed estimator in practice. The first modification is the parallel execution. Recall that the general procedure of wire 
cutting proceeds sequentially, which results in high run-time costs in practice. To address this issue, we parallelize all of $r$ subcircuits by enumerating all possible cutting variables $\pbra{\bm{z}^{(l)}, U_{\bm{z}^{(l)}}, p_{\bm{z}^{(l)}}}$, intermediate outcomes $\bm{c}^{(l)}$, and reinitialized states $\kappa_i$ for $l \in [k_2]$. This approach remains efficient when both $k_1$ and $k_2$ are chosen as small constants.
The second modification is an error mitigation technique, which improves the reliability of the output distributions when we intend to suppress the influence of measurement noise. We postpone the relevant algorithmic implementation and theoretical results about the complexity in SI~E.

\noindent\textbf{Remark}. For completeness, we also prove the computational separation between the cross-platform verification without circuit cutting and the state tomography in terms of the sample complexity. In particular, prior works only numerically exhibited that cross-platform verification using local random unitaries generally requires $\lesssim  2^n$ copies of quantum states if $\rho =\sigma$. In SI~C,  we prove that the required number of measurements scales as $\sqrt{(18/5)^n}$ when $\rho = \sigma$, and as $\sqrt{6^n}$ when $\rho \neq \sigma$. These results establish a rigorous separation between cross-platform verification and quantum state tomography, where the latter requires  $\Omega(2^{2n})$ copies with single-copy measurements \cite{haah2016sample,nayak2025lower}.

\subsection{Federated quantum kernels} Beyond cross-platform verification in modular quantum processors, our approach can also support trustworthy quantum kernel machines~\cite{schuld2019quantum,Havl2019supervised,peters2021machine}. Denote $\mathcal{D}=\{\bm{x}^{(i)}, y^{(i)}\}_{i=1}^N$ as the given dataset with $N$ training examples. As with classical kernels, the quantum kernel matrix $K$ is an $N\times N$ matrix, where its entries quantify the similarity between examples in the quantum feature space. Mathematically, $K_{ij}=\braket{\psi(\bm{x}^{(i)})|\psi(\bm{x}^{(j)})}$ for $\forall i, j\in [N]$, where the state $|\psi(\bm{x}^{(i)})\rangle$ is prepared by applying a tailored parameterized circuit $U(\bm{x}^{(i)})$ to a fixed initial state $\ket{0}^{\otimes n}$. Despite theoretical advances in certain  tasks~\cite{huang2021power,liu2021rigorous,glick2024covariant}, the construction of $K$ raises concerns regarding trustworthiness, since its computation requires all data in $\mathcal{D}$ to be exchanged among different entities, which is undesirable in sensitive domains such as healthcare. 

Given that $K_{ij}$ amounts to state similarity estimation, the cross-platform nature of our approach ensures that all entries in $K$ can be completed instead of sharing raw data. This federated manner addresses the privacy-preserving issue. Furthermore, when the underlying circuit $U(\bm{x}^{(i)})$ is adapted to modular quantum processors, the required sample complexity is substantially reduced compared to the single-processor setting, ensuring computational efficiency.

\section{EXPERIMENTAL RESULTS} 
\subsection{Similarity Estimation of GHZ states}
We experimentally validate the proposed cross-platform verification protocol using GHZ states distributed across modular regions of a superconducting quantum processor. In our setup, a central qubit is cut to divide the GHZ state into two modular components, i.e., $r = 2$, $k_1 = 1$, and $k_2 = 1$. The experiment involves $19$ frequency-tunable and asymmetric qubits arranged in a square lattice, and the modular architecture is emulated by partitioning the chip into spatially separated regions (refer to SI~D and E for the details).

In the first experiment shown in Fig.~\ref{fig:experiment_fid_main}$\bm{a}$, we evaluate $9$-qubit GHZ states prepared across five distinct regions (D1–D5). Each region comprises a subset of the chip, and cross-platform fidelities are estimated pairwise between regions using our wire-cutting protocol. The corresponding modular layouts are illustrated above the matrix, and fidelity values are shown in the off-diagonal entries. We choose $N = 50$ with 1000 snapshots and repeat the procedure across 12 rounds with slight variations to obtain the statistical results. Refer to SI~E for the detailed settings. The achieved results demonstrate that GHZ states prepared in different regions exhibit high mutual fidelity, where all of the results are above 0.62. 

In Fig.~\ref{fig:experiment_fid_main}$\bm{b}$, we further examine GHZ states varying from $5$ to $11$ qubits partitioned into two modular components with equal sizes: $3+3$, $4+4$, $5+5$, and $6+6$ qubits. The corresponding modular configurations are depicted above each bar group. For each layout, we report the cross-platform fidelity relative to ideal GHZ states, both with and without the proposed measurement error mitigation technique. The experimental settings follow a similar structure to the $5+5$-qubit case employed in Fig.~\ref{fig:experiment_fid_main}$\bm{a}$. That is, for the $4+4$ and $6+6$ configurations, we use $N = 40$ and $N = 66$, respectively, with $1000$ snapshots per random unitary. For the $3+3$-qubit case, we enumerate all possible unitaries, totaling 9 and 27 for the two isolated parts, respectively, with 500 snapshots taken per unitary. The achieved results show that when error mitigation is employed, all fidelities in the $6+6$-qubit case relative to the ideal GHZ states exceed 0.79, attaining a substantial improvement compared to the results without error mitigation. This indicates the effectiveness of the proposed error mitigation technique.

To verify that the reconstructed quantum states match the target GHZ states beyond fidelity estimates, we perform Pauli tomography on $5$-qubit GHZ states prepared across two modules ($\mathsf{Q}$-$\mathsf{A}$ and  $\mathsf{Q}$-$\mathsf{B}$). To ensure that the tomography state corresponds to the same quantum state used in the cross-platform fidelity estimation experiment, the Pauli tomography results are obtained using the same measurement data from Fig.~\ref{fig:experiment_fid_main}\textbf{b}, but with different post-processing. As shown in Fig.~\ref{fig:experiment_fid_main}$\bm{c}$, we compare 64 experimentally estimated Pauli expectation values with their theoretical predictions. The mean squared errors (MSE) of the tomography values using the first 64 Pauli operators and the full set of 1024 are $0.00345$ and $0.0043$ for Platform 1, and $0.00162$ and $0.0023$ for Platform 2, respectively. The lower MSE for the first 64 operators reflects their predominantly local structure, suggesting that expectation values of global Pauli operators are more susceptible to noise.
The agreement across all Pauli terms further confirms the accuracy of our method in reconstructing global quantum states from modular subcircuits. Refer to SI~E for the omitted implementation details and more experimental results.

\begin{figure*}[t]
    \centering
    \includegraphics[width=1.0\linewidth]{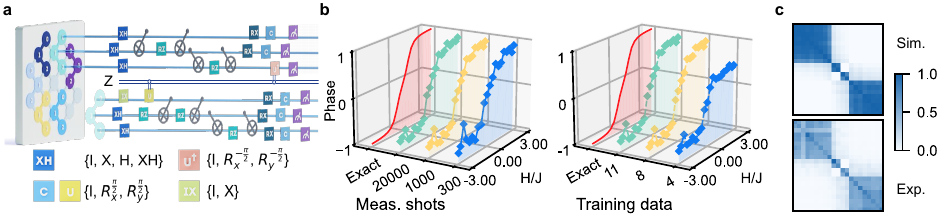}
    \caption{Experimental phase learning using the cross-platform verification in circuit-cutting case. \textbf{a.} The parameterized quantum circuit used to compute elements of the kernel matrix. \textbf{b.} Phase learning performance as a function of the number of measurement shots with a fixed training set of 8 data points, and as a function of training set size with 1000 shots. \textbf{c.} Comparison between the simulated kernel matrix and the experimentally obtained kernel matrix.}
    \label{fig:application_phaseLearn}
\end{figure*}

\subsection{Federated quantum kernels for phase identification}
We next evaluate the performance of the federated quantum kernels by applying the proposed models to quantum phase learning tasks, a compelling setting for demonstrating quantum advantage, where the underlying problem is believed to be classically intractable~\cite{wu2023quantum, bouland2019complexity}. In this experiment, a total of 18 frequency-tunable, asymmetric qubits are used.
In particular, the phase structure under investigation is derived from a specific family of transverse-field Ising models, i.e.,   
\begin{align}
  \Fcal=\Big\{H(h) \big| H(h) = - \frac{1}{2}\sum_{j=1}^{n-1} Z_jZ_{j+1} - h\sum_{j=1}^n X_j \Big\},
\end{align}
where $h\in  [-1.45, 1.45]$ refers to the strength of the external field and the number of qubits is set to $n=5$. Let the ground state of $H(h)$ be $\rho(h)$. The object is to build a quantum kernel machine that can accurately predict the phase of $\rho(h')$ for any unseen $h'$. 

To complete this learning task, we first construct the training dataset $\mathcal{T}=\{\rho(h^{(i)}), y^{(i)}\}$, where $\rho(h^{(i)})$ and $ y^{(i)}$ denote the ground state and the quantum phase for the $i$-th Hamiltonian $H(h^{(i)})$, respectively. The preparation of $\rho(h^{(i)})$ is completed by a variational quantum eigen-solver with a Hamiltonian-informed ansatz (see SI~E for the implementation detail). Once the dataset $\mathcal{T}$ is prepared, we apply it to construct the proposed federated quantum kernel $K$, where the entry $K_{ij}$ amounts to the estimation value of $\tr{\rho(h^{(i)}) \rho(h^{(j)})}$ returned by the proposed cross-platform verification algorithm. Here we consider the settings with $k_1=1$, $k_2=1$, and $r=2$, where the circuit implementation is depicted in Fig.~\ref{fig:application_phaseLearn}\textbf{a}. Given access to the established federated quantum kernel $K$, the kernel support vector regression (SVR) model is utilized to conduct the prediction~\cite{scholkopf2002learning,cherkassky2004practical,smola2004tutorial}.

To evaluate how the performance of the proposed federated quantum kernel depends on the number of training examples and the number of snapshots used to estimate each entry of the kernel matrix $K$, we vary the training set size from $4$ to $11$ examples and the number of snapshots from $300$ to $2000$. Prediction performance is assessed on 21 test examples drawn from the same distribution as the training data, as shown in Fig.~\ref{fig:application_phaseLearn}\textbf{b}. When the number of training examples is $8$ and the snapshots yield $2000$, the average mean squared error between the predicted and exact test labels is $0.026$, validating the effectiveness of the proposed method. 

We further examine the consistency between the experimentally reconstructed kernel matrix and the kernel matrix obtained through classical estimation. All hyperparameter settings are identical to those used in the previous task, with  $k_1=k_2=1$, $r=2$, and 1000 snapshots used to approximate the $21 \times 21$ kernel matrix constructed from 21 selected training data points. Using centered kernel alignment as the evaluation metric~\cite{kornblith19similarity}, the achieved core is $0.98$, indicating a high degree of structural similarity between the experimentally reconstructed kernel matrix and the target kernel matrix.

\section{CONCLUSION}
In this study, we develop an algorithm for cross-platform verification tailored to modular quantum processors connected via wire cutting. Theoretical analysis establishes how the estimation error scales with the number of measurements, and in certain regimes, the proposed method outperforms existing approaches that do not incorporate circuit cutting, enabling verification at larger scales. Building on this framework, we design a federated quantum kernel to address trust and privacy concerns in quantum machine learning. Complementing the theoretical contributions, we perform systematic experiments on a superconducting quantum processor. Our results demonstrate successful verification of cross-platform fidelity for GHZ states involving up to $11$ qubits, and accurate phase classification in the transverse-field Ising model, validating the practical utility of the proposed methods for scalable quantum information processing.

Several important research directions warrant further exploration. One direction is to investigate whether incorporating intermediate measurements in a quantum device can extend the applicability of our algorithm to more general settings~\cite{CarreraVazquez2024Combing,decross2023qubit}. At the same time, advancing hardware capabilities to support such measurements is essential, as implementing intermediate measurements remains challenging on current superconducting quantum platforms~\cite{swiadek2024enhancing}. Another promising direction is to establish lower bounds for cross-platform verification in the presence of circuit cutting. Understanding the extent of this separation is a key open question. Finally, it is of interest to develop cross-platform verification algorithms for modular quantum processors when limited quantum communication is allowable. Any progress in this setting could offer valuable insights into the verification of distributed quantum computations.

\section*{ACKNOWLEDGEMENT}
This work was supported by the National Natural Science Foundation of China (12405014, 12204228), the Innovation Program for Quantum Science and Technology (2021ZD0301703), and Guangdong Basic and Applied Basic Research Foundation (2024A1515011714, 2022A1515110615).

\newpage
 
\clearpage

\onecolumngrid

\appendix 

\begin{center}
	{\textbf{\large{Supplementary Information for: ``State Similarity in Modular Superconducting Quantum Processors with Classical Communications''}}}
\end{center}

\tableofcontents

\newpage

\renewcommand{\appendixname}{SI}
 \renewcommand\thefigure{\thesection.\arabic{figure}}   
  \renewcommand\thetable{\thesection.\arabic{table}}   
 
\renewcommand{\figurename}{Supplementary Figure}
\renewcommand{\tablename}{Supplementary Table}

\bigskip

\renewcommand{\theequation}{S\arabic{equation}}
\setcounter{equation}{0}

In this Supplementary Information (SI), we provide detailed descriptions of the algorithms, theoretical analyses, experimental setups and proofs supporting our main results. SI~\ref{supp:preliminary} introduces the necessary notations and background knowledge used throughout the remainder of the supplementary information. In SI~\ref{supp:complexity_distance_based_alg}, we analyze the theoretical complexity of an existing distance-based cross-platform algorithm, which forms the basis for our subsequent proof. SI~\ref{supp:algorithm_details} contains details of our proposed algorithm along with the corresponding proofs. SI~\ref{supp:intro_supconduct_device} presents information about the experimental device. 
 Finally, SI~\ref{supp:experimental_res} provides additional experimental results, including parameter settings and simulation data.

\section{Preliminary}
\label{supp:preliminary}
We introduce the Pauli-transfer matrix representation (Liouville representation) for operators and superoperators~\cite{wood2011tensor}. A nonzero linear operator $A$ acting on a Hilbert space can be represented as a vector in the Hilbert-Schmidt space of dimension $d^2$ on the basis of Pauli operators, where $d$ is the dimension of the Hilbert space. Specifically, we define the rescaled operator  $\superket{A}$ as $\superket{A} := A/\sqrt{\tr{AA^\dagger}}$. Using this vectorized representation, the operator $A$ is expanded as a linear combination of the Pauli basis: $\superket{A} = \sum_{j=0}^{d^2-1} \superbraket{P_j|A} \superket{P_j}$, where $\cbra{\superket{P_j}}_{i=1}^{d^2 - 1}$ represented as the vectorized Pauli basis, with each Pauli operator rescaled by $\superket{P_j}:=P_j/\sqrt{d}$,
and the inner product $\langle\!\langle B|A \rangle\!\rangle := \tr{AB^\dagger}/\sqrt{\tr{AA^\dagger}\tr{BB^\dagger}}$. 
We denote the map of a linear operator as $\Ecal\superket{A} = \Ecal(A)/\sqrt{\tr{AA^\dagger}}$.
Throughout this SI, we denote $\Fbb$ as the $2n$-qubit SWAP gate, which exchanges the $j$-th and $(j+n)$-th qubits for $j = 1, \dots, n$. We use $n^{\underline{t}}$ to denote the falling factorial, defined as
$n^{\underline{t}} = n \times (n - 1) \times \cdots \times (n - t + 1)$.

The channel representation of a unitary operation is defined as $\Ucal(\cdot):= U(\cdot) U^\dagger$. We utilize $\Clgroup_k$ to denote the $k$-qubit Clifford group. We use $\Ibb$ to denote the associated system's identity operator without specifying the dimension. The notation $\Ical(\cdot)$ denotes the identity channel. To specify its operation on a $d$-dimensional Hilbert space, we use $\Ibb_d$ and $\Ical_d$, respectively.
We use $D(\bm{x}, \bm{y})$ to denote the distance between two bitstrings $\bm{x}, \bm{y} \in \cbra{0,1}^n$, defined as $D(\bm{x}, \bm{y}) = \sum_i x_i y_i$. We use $\# \bm{a}$ to denote the number of occurrences of $\bm{a}$ in the set of all measurement shots.

\subsection{Cross-platform fidelity estimation}

Let $\rho$ and $\sigma$ be two quantum states prepared by two quantum platforms, denoted by Platform 1 and Platform 2,  respectively. The cross-platform fidelity~\cite{wang2008alternative,liang2019quantum, brydges2019probing, Elben2020cross} between the states prepared by Platform 1 and Platform 2 is
\begin{align}\label{eqn:def-cross-fide}
    \Fcal(\rho, \sigma) = \frac{\tr{\rho \sigma}}{\sqrt{\tr{\rho^2}\tr{\sigma^2}}},
\end{align}
which is the inner product of $\rho$ and $\sigma$ normalized by the geometric mean of their purities $\tr{\rho^2}$ and $\tr{\sigma^2}$. This quantity is calculated by separately applying quantum operations on the individual platform. To date, two representative classes of algorithms have been proposed that utilize random measurements to estimate $  \Fcal(\rho, \sigma)$, i.e., the ``distance-based algorithm''~\cite{Elben2020cross, brydges2019probing} and the ``collision-based algorithm''~\cite{Anshu2022Distributed}. In what follows, we briefly review their mechanisms.

\subsubsection{Distance-based algorithm for cross-platform fidelity estimation}
In distance-based algorithm~\cite{Elben2020cross, brydges2019probing}, the estimator for $\tr{\rho \sigma}$ in Eq.~(\ref{eqn:def-cross-fide}) is defined as
\begin{equation}
    \bar{v} := \frac{2^n}{N}\sum_{t=1}^N \sum_{\bm s, \bm s'} (-2)^{-D(\bm s, \bm s')}\hat{p}_{\bm s}(\Wcal_t)\hat{q}_{\bm s'} (\Wcal_t)
    \label{eq:cross_platformEstimator}
\end{equation}
 where $\hat{p}_{\bm s} = {\#\bm s}/{m}$ represents the frequency of $\bm s$ occurring in $m$ repetitions for Platform $1$ with the random unitary be chosen as $\Wcal_t$, $\hat{q}_{\bm s}$ follows the similar definition for Platform $2$, and $N$ random unitaries $\Wcal_1,\ldots, \Wcal_N$ are uniformly sampled from an $n$-qubit global (or tensor product of local) Clifford group. 
By the twirling property of the 2-design of Haar measure~\cite{mele2024introduction}, it can be proven that
\begin{equation}
\tr{\rho \sigma} \equiv \Ebb_{\Wcal,\hat{p}_{\bm s}, \hat{q}_{\bm s'} }[ \bar{v}] = 2^n \sum_{\bm s, \bm s'} (-2)^{-D(\bm s, \bm s')}\Ebb_{\Wcal,\hat{p}_{\bm s}\hat{q}_{\bm s'} } \sbra{\hat{p}_{\bm s}(\Wcal) \hat{q}_{\bm s'}(\Wcal)}.
\label{eq:cross_platFid}
\end{equation}

Numerical studies~\cite{Elben2020cross} have shown that, for almost all quantum states, the required number of measurements scales as $2^{bn}$, with $b  \lesssim 1$ in general. More specifically, for pure (entangled) Haar random states, the scaling exponent is found to be $b = 0.6 \pm 0.2$, whereas for pure product states, it is $b = 0.8 \pm 0.1$~\cite{Elben2020cross}. However, a rigorous theoretical foundation for these numerical observations remains lacking.
In SI~\ref{supp:algorithm_details} we provide a theoretical analysis of the measurement requirements to address this gap.

\subsubsection{Collision-based algorithm for cross-platform fidelity estimation}
Anshu et al.~\cite{Anshu2022Distributed} proposed a collision-based algorithm that bears similarity to the distance-based cross-platform algorithm, where both of them require the application of a random unitary before measurement. Different from the distance-based approaches, they prove that the number of measurements required to achieve an estimation error of $\varepsilon$ with high probability scales with  $\Theta\pbra{\sqrt{2^n}/\varepsilon}$. However, a key limitation of their approach is that it requires a global random unitary acting on all qubits, which often requires substantial and complicated quantum operations.

The algorithmic implementation proposed by Ref.~\cite{Anshu2022Distributed} is as follows. Denote the sampled unitary sequence as $\Wmat = \cbra{\Wcal_1,\ldots,\Wcal_{N}}$, where the $i$-th unitary $\Wcal_i$ is randomly sampled in $n$-qubit Haar measure, and $d = 2^n$. Let $m$ be the number of measurement shots for each fixed circuit configuration. Then, the unbiased estimator of $\tr{\rho \sigma}$ takes the form as
\begin{equation}
\bar{v} = \frac{(d+1)}{N m^2}\sum_{\Wcal \in \Wmat} \sum_{i,j=1}^m \mathbbm{1}\sbra{\bm s_i^{\Wcal} = {\bm s'}_j^{\Wcal}} - 1,
\end{equation}
where $\mathbbm{1}[a=b]$ amounts to the indicator function, $\bm s_i^{\Wcal}$ (or ${\bm s'}_j^{\Wcal}$) is the measured bitstring associated with unitary $\Wcal\in \Wmat$ for the quantum state $\rho$ (or $\sigma$). 
In addition, let the set of measurement outcomes in two platforms be $\Xmat = \{\bm s_i^{\Wcal_j}\}_{i,j}$. The variance of estimator $v(\Wmat, \Xmat)$ is bounded by
\begin{align}
\var\pbra{v(\Wmat, \Xmat)} = \Ocal\pbra{\frac{1}{Nd} + \frac{d}{N^2 m^2} + \frac{1}{Nm}}.
\label{eq:variance_collision_based_CP}
\end{align}
To ensure the variance in Eq.~\eqref{eq:variance_collision_based_CP}, the random unitaries $\cbra{\mathcal{W}_j}_{j=1}^N$ must be sampled from at least a 4-design and cannot be substituted with a tensor product of a local unitary group. This requirement poses significant implementation challenges given current quantum resources. Consequently, in the SI~\ref{supp:complexity_distance_based_alg}, we focus exclusively on optimizing the distance-based algorithm with classical communication.

\subsection{Preparation of the quantum state with circuit cutting}
Lowe et al.~\cite{Lowe2023fast} proposed the wire-cutting algorithm by introducing the random measurement technique.
The core idea is to introduce two channels $\Phi_0(\cdot)$ and $\Phi_1(\cdot)$ operating on $k$ qubits, such that
\begin{equation}
\Phi_0 = \Ebb_{\Ccal}\sbra{\Ccal^\dagger\circ M_{\bm s} \circ \Ccal},
\label{eq:phi_0}
\end{equation}
where $\Ccal$ is a randomly chosen Clifford circuit from $\Clgroup_k$, $M_{\bm s} = \sum_{\bm s} \superket{\bm s}\superbra{\bm s}$, and
\begin{equation}
    \Phi_1 = \superket{\Ibb}\superbra{\Ibb}.
    \label{eq:phi_1}
\end{equation}
By the twirling property of the Clifford group, let $z\in\cbra{0,1}$ be the Bernoulli random variable such that $\Pr[z = 0] = \frac{2^k}{2^{k+1}+1}$, where $k$ is the number of qubits of the channel $\Phi_0$ and $\Phi_1$. Then we can generate an identity channel
\begin{equation}\label{SI:eqn:wire-cut}
    \Ical_{2^k} = (2^{k+1}+1)\Ebb_z\sbra{(-1)^z \Phi_z}.
\end{equation}

In addition, supported by the corollary of Schur's lemma~\cite{fulton2013representation}, we observe that
\begin{equation}
 \Phi_0 = \Ebb_{\Ccal}\sbra{\Ccal\circ M_{\bm s} \circ \Ccal^{\dagger}} = \superket{\Ibb}\superbra{\Ibb} + \frac{1}{2^{k}+1} \Pi_1,  
 \label{eq:phi_expect_formula}
\end{equation}
where $\Pi_1=\sum_{P\in \cbra{\Ibb_2,X,Y,Z}^n\backslash \Ibb_{2^n}} \superket{P}\superbra{P}$.

We next illustrate how to apply the wire-cutting algorithm to the circuit depicted in Supplementary Figure~\ref{suppfig:cutting_technique}\textbf{a}, where an $n$-qubit circuit is partitioned into \textsf{Q}-\textsf{A} and \textsf{Q}-\textsf{B} with the number of cutting qubits being $k=\abs{A\cap B}$. The resulting subcircuits for \textsf{Q}-\textsf{A} and \textsf{Q}-\textsf{B} contain $n_A = k + k_A$ and $n_B = k + k_B$ qubits, respectively.   

Following the Liouville representation, the explored $n$-qubit state takes the form as
\begin{equation}
 \superket{\hat{\rho}}  = (\Ical_{2^{k_A}} \otimes \Ucal_B) \mathcal{I}_{2^k} (\Ucal_A \otimes \Ical_{2^{k_B}}) |0^n\rangle\!\rangle.
\end{equation}
Then, supported by Eq.~\eqref{eq:phi_expect_formula}, we have
\begin{equation}
  \superket{\hat{\rho}} = \mathbb{E}_{z}[\superket{\hat{\rho}_z}] \quad \text{with} \quad  \superket{\hat{\rho}_z} = (\Ical_{2^{k_A}} \otimes \Ucal_B)(2^{k+1} + 1)(-1)^z (\Ical_{2^{k_A}} \otimes \Phi_z \otimes \Ical_{2^{k_B}})(\Ucal_A \otimes \Ical_{2^{k_B}}) |0^n\rangle\!\rangle. 
   \label{eq:def_rho_z}
\end{equation}
It can also be easily shown that the coefficients before $\Phi_0$ and $\Phi_1$ are minimized when $\Pr[z = 1] = \frac{2^k}{2^{k+1}+1}$, yielding the smallest variance for the estimator.

Supplementary Figure~\ref{suppfig:cutting_technique}\textbf{b} visualizes the circuit after cutting. Sometimes we omit  $\Ical$ terms or the subscripts of $\Ical$ for simplicity, with a slight abuse of the notation. Note that we can only obtain an estimator of $\superket{\hat{\rho}_z} $ due to the randomness generated by the channels $\Phi_0$ and $\Phi_1$. We can formally define $\superket{\hat{\rho}_z} $ as  
\begin{align}
\superket{\hat{\rho}_z} = \begin{cases}
\Ucal_B(2^{k+1} + 1)\Ccal \frac{\superket{\bm c}_{\bm e}\superbra{\bm c}_{\bm e}\Ccal^\dagger \Ucal_A\superket{0^n}}{\superbra{\bm c}_{\bm e}\text{Tr}_{\bar{\bm e}}\pbra{\Ccal^\dagger \Ucal_A\superket{0^n}}}, &\text{if } z = 0\\
\Ucal_B (2^{k+1} + 1)(-1)
\frac{\superket{\bm c'}_{\bm e} \superbra{\bm c}_{\bm e}\Ccal^\dagger\Ucal_A\superket{0^n}}{\superbra{\bm c}_{\bm e }\text{Tr}_{\bar{\bm e}}\pbra{\Ccal^\dagger \Ucal_A\superket{0^n}}}
 , &\text{otherwise,}
\end{cases}
\label{eq:def_rho_z_approx}
\end{align}
where $\mathcal{C}$ is randomly sampled from the $k$-qubit Clifford group and acts on the set of cutting qubits $\bm{e} = A \cap B$. The reduced density matrix $\text{Tr}_{\bar{e}}(\rho)$ corresponds to qubits $\bm{e}$ after tracing out the remaining qubits. The bit-string $\bm{c}$ denotes the measurement outcomes on $\bm{e}$ after applying $\mathcal{U}_A$ to $\superket{0^{n_A}}$, while $\bm s$ represents the outcomes obtained by inserting $\mathcal{C}_{\bm e}^\dagger$ before measurement. The bit-string $\bm{c}'$ is independently sampled from the uniform distribution.

\begin{figure}
    \centering
    \includegraphics[width=0.8\linewidth]{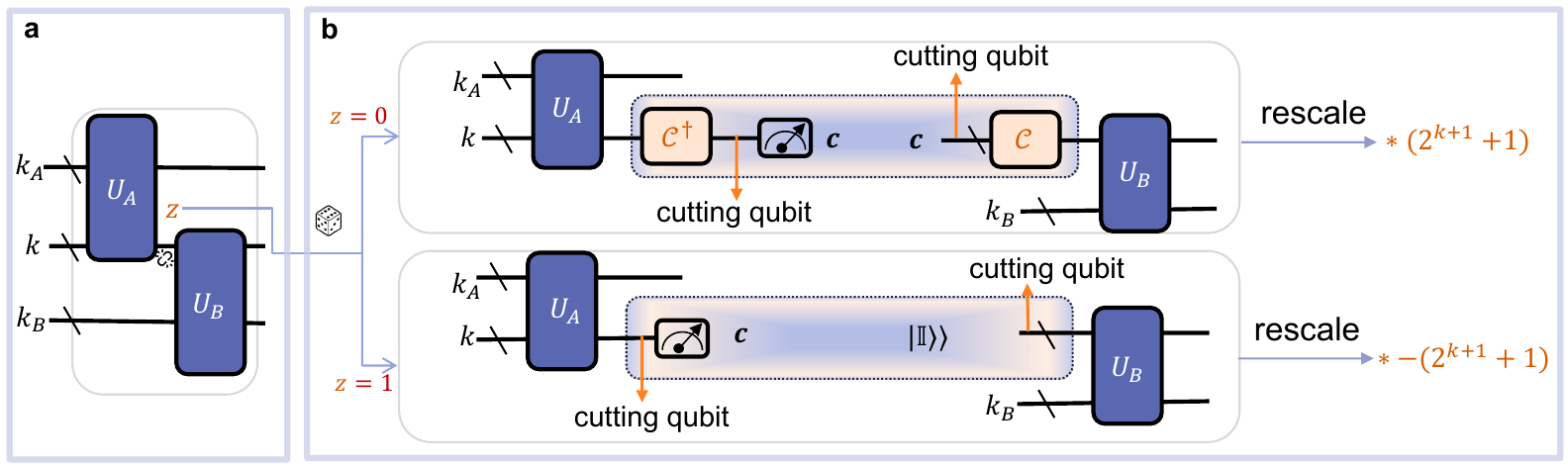}
    \caption{Schematic illustration of the circuit-cutting algorithm. We assume the cut involves $k$ qubits. \textbf{a.} The circuit before the cut. \textbf{b.} The sub-circuits after the cut. When the associated Bernoulli random variable equals $0$, the upper circuit is executed; when it equals $1$, the lower circuit is used. In the latter case, the inputs to the cut qubits are uniformly randomly selected from $\{0,1\}^k$ following measurement.
}
    \label{suppfig:cutting_technique}
\end{figure}

\subsection{Overview of quantum kernel}

Machine learning techniques have become pivotal in the era of big data, demonstrating significant potential in both classical computing~\cite{lecun2015deep, lin2017does} and quantum computing~\cite{carleo2017solving, van2017learning, cong2019quantum, du2025quantum}.
Among these techniques, kernel methods play a crucial role by enabling the application of linear algorithms to nonlinear problems through the use of feature maps. In the context of quantum computing, quantum kernel methods have emerged as a powerful tool in machine learning~\cite{Havl2019supervised,huang2021power,wang2021towards,peters2021machine,kubler2021inductive,glick2024covariant}. These methods involve mapping classical input data into a high-dimensional quantum feature space, allowing linear learning algorithms to be effectively applied to complex, non-linear data structures. Specifically, classical input data $\bm{x}$ is embedded into quantum states via a parameterized quantum circuit $U(\bm{x})$, and the quantum kernel is defined as the inner product between these quantum states.

For instance, in our setting, we encode the strength of the transverse magnetic field $h$ in the 1D Transverse Field Ising Model into the corresponding ground state $\ket{\psi(h)}=U(h)\ket{0}^{\otimes n}$, where $U(h)$ consists of tunable single-qubit gates and fixed quantum gates with a predefined gate layout. Given $R$ data points $\{h^{(i)}, y^{(i)}\}_{i=1}^R$ with $y^{(i)}$ being the label for the $i$-th example, we define an $R \times R$ quantum kernel matrix $K$ with entries  
\begin{align}
  K_{ij} = \braket{\psi(h^{(i)})|\psi(h^{(j)})}.
\end{align}

The resulting kernel matrix can be further processed using classical machine learning models such as Support Vector Machines (SVM), Support Vector Regression (SVR)~\cite{scholkopf2002learning, cherkassky2004practical, smola2004tutorial}, or other kernel-based algorithms~\cite{wu2023quantum} to perform learning and inference. In our experiments, we employ the SVR model, which extends the principles of SVM to regression tasks by finding a function that deviates from the true targets by no more than a specified margin while maintaining model simplicity.

\section{Complexity analysis of Distance-based Cross-platform fidelity estimation algorithm}
\label{supp:complexity_distance_based_alg}
Here, we provide a theoretical analysis of the sampling complexity of the distance-based algorithm introduced in SI~\ref{supp:preliminary}, addressing a gap in previous studies that focused solely on numerical results.

Let $\Xi^{(k)}= \Ebb_{\Wcal\in \mu_H}\sbra{\Wcal^{\otimes k}}$ be the $k$-th moment operator with respect to the probability measure $\mu_H$, where $\mu_H$ is defined as the Haar measure on the unitary group.
\begin{lemma}[Ref.~\cite{mele2024introduction}]
The average channel of the $2$-th moment of Haar measure  $\Xi^{(2)}:=\Ebb_{\Wcal\sim \mu_H} \sbra{\Wcal^{\otimes 2}}$ satisfies 
\begin{align}
    \Xi^{(2)}{\superket{\bm{s_1}}\otimes\superket{\bm{s_2}}}=\begin{cases}
        &\frac{1}{d^2-1}\Ibb + \frac{\mathbb{F}}{d(1-d^2)}\text{ if }\bm{s_1}\ne \bm{s_2},\\
        &\frac{1}{d(d+1)}\pbra{\Ibb + \mathbb{F}}\text{ otherwise,}
    \end{cases}
\end{align}
where $\mathbb{F}$ is the $2n$-qubit SWAP gate which exchanges the $j$-th and $(j+n)$-th qubits for $j = 1, \dots, n$, and each $\bm{s_l}\in \cbra{0,1}^{n}$ for $l\in [2]$.
\label{lem:Haar_second}
\end{lemma}

\begin{lemma}[Refs.~\cite{mele2024introduction,collins2006integration,collins2003moments,brandao2016local}]
The average channel of the $t$-th moment of Haar measure $\Xi^{(t)}:=\Ebb_{\Wcal\sim\mu_H}\Wcal^{\otimes t}(\cdot)$ satisfies
\begin{equation}
    \superbra{A}\Xi^{(t)}\superket{B} = \sum_{\zeta,\tau\in S_t} C_{\zeta,\tau} \tr{A V_{\zeta}} \tr{V_{\tau}^\dagger B},
\end{equation}
for any linear operators $A,B \in \Lcal(\Hcal_d)$, where $C$ is the pseudo-inverse of $G$, and $G$ is a Gram matrix with $G_{\zeta,\tau} = d^{\text{cycles}(\zeta^{\dagger} \tau)}$, $\text{cycles}(\tau)$ denotes the number of cycles in the permutation $\tau$, $V_{\zeta}, V_{\tau}$ are two permutation matrices corresponding to $\zeta,\tau\in S_{t}$.
\label{lem:Haar_four}
\end{lemma}

\begin{proposition} 
Following the notation in Lemma~\ref{lem:Haar_four}, we have $\sum_{\tau\in S_t} C_{\zeta, \tau} = \frac{1}{(d+t-1)^{\underline{t}}}$ for any fixed permutation $\zeta$.  
\label{pro:Haar_coefs}
\end{proposition}
\begin{proof}

By Lemma~\ref{lem:Haar_four}, we have
\begin{align}
    \Xi^{(t)}\superket{\bm 0}^{\otimes n} = \sum_{\zeta \in S_t} \sum_{\tau\in S_t}C_{\zeta,\tau}V_{\zeta}=  \sum_{\zeta \in S_t} C_{\zeta}V_{\zeta},
\end{align} 
where the second equality holds by letting $C_{\zeta} := \sum_{\tau} C_{\zeta,\tau}$, which is invariant under the transformation of $V_\tau$. Then we obtain
\begin{align}
  \sum_{\zeta} C_{\zeta} V_{\zeta} = \Xi^{(t)}\superket{\bm 0}^{\otimes n} &=   V_{\tau}\Xi^{(t)}\superket{\bm 0}^{\otimes n}\\
  &= V_{\tau} \sum_{\zeta}  C_{\zeta}  V_{\zeta}\\
  &=\sum_{\tau\zeta}C_{\tau\zeta} V_{\zeta}.
\end{align}

Hence $C_{\tau\zeta} = C_{\zeta}=C_{\Ibb}$ for any $\zeta, \tau$. By taking the trace, we have $C_{\Ibb} = \frac{1}{t!{d+t-1\choose t}} = \frac{1}{(d+t-1)^{\underline{t}}}$.  
\end{proof}

The following lemma serves as a key component in the proof of our main theorem. Here, $s, s', b, b'$ denote the measurement outcomes for a single qubit. This lemma is employed to bound the variance of the cross-platform fidelity when local random unitaries are applied.
\begin{lemma}
\begin{equation}
    \sum_{\tau \in S_4} \sum_{ s, s',b,  b'\in\cbra{0,1}} (-2)^{-D(s,  s') - D(b,b')} \tr{V_{\tau} \superket{s} \superket{s'} \superket{b} \superket{b'}} = 36,
\end{equation}
where $V_{\tau}$ is the permutation operator associated with permutation $\tau\in S_4$.
\label{lem:intermediate_result_crossplatform}
\end{lemma}
\begin{proof}
We enumerate all combinations of $s, s', b, b' \in \cbra{0,1}$ and list the corresponding values of
$y_1 := (-2)^{-D(s, s') - D( b, b')}$ and
$y_2 := \sum_{\tau \in S_4} \tr{V_{\tau} \superket{s} \superket{s'} \superket{b} \superket{b'}}$
in Supplementary Table~\ref{supptab:intermediate_proof_tab}. Using these tabulated values, we immediately obtain the following result: 
\begin{equation}
\sum_{\tau \in S_4} \sum_{s, s',b, b' \in \cbra{0,1}} (-2)^{-D( s, s') - D( b, b')} \tr{V_{\tau} \superket{s} \superket{s'} \superket{b} \superket{b'}} = 36. 
\end{equation}

\begin{table}[ht]
\centering
\caption{Enumerate the values of $y_1$ and $y_2$ for all combinations of $s, s', b, b'$ in $\cbra{0,1}$.}
\begin{tabular}{|c |c | c | c | c | c|}
\hline
$s$ & $s'$ & $b$ & $b'$ & $y_1 := (-2)^{-D( s, s') - D(b, b')}$ & $y_2 := \sum_{\tau \in S_4} \tr{V_{\tau} \superket{s} \superket{s'} \superket{b} \superket{b'}}$\\
\hline
 0 & 0 & 0 & 0 & \multirow{3}{*}{1} & $4!$ \\
\cline{1-4} \cline{6-6} 
0 & 0 & 1 & 1 &  & $4$ \\
\cline{1-4} \cline{6-6} 
1 & 1 & 0 & 0 &  & $4$ \\
\cline{1-4} \cline{6-6} 
1 & 1 & 1 & 1 & & $4!$ \\
\hline
 0 & 0 & 0 & 1 & \multirow{8}{*}{$-\frac{1}{2}$} & \multirow{8}{*}{$3!$} \\
\cline{1-4}  
 0 & 0 & 1 & 0 &  & \\
\cline{1-4} 
 1 & 1 & 0 & 1 & & \\
\cline{1-4} 
 1 & 1 & 1 & 0 &  & \\
\cline{1-4} 
 0 & 1 & 0 & 0 &  & \\
\cline{1-4} 
 1 & 0 & 0 & 0 &  &  \\
\cline{1-4} 
 0 & 1 & 1 & 1 &  & \\
\cline{1-4} 
 1 & 0 & 1 & 1 & &  \\
 \hline
 0 & 1 & 0 & 1 & \multirow{4}{*}{$-\frac{1}{4}$} & \multirow{4}{*}{$4$} \\
\cline{1-4} 
 0 & 1 & 1 & 0 & & \\
\cline{1-4} 
 1 & 0 & 1 & 0 & &  \\
\cline{1-4} 
 1 & 0 & 0 & 1 &  &  \\
\hline
\end{tabular}
\label{supptab:intermediate_proof_tab}
\end{table}
\end{proof}

In the following, we prove that the variance of the estimator defined in Eq.~\eqref{eq:cross_platformEstimator}
can be bounded to $\frac{1}{N} \pbra{c + \frac{c'}{m^2}(5/2)^n}$ when the random unitary $\Wcal_t$ is sampled from a 4-design. The following theorem demonstrates that the global distance-based cross-platform algorithm may underperform compared to the collision-based algorithm.
\begin{theorem}
The number of measurements can be bounded by $\Ocal\pbra{\sqrt{(5/2)^n}/\varepsilon^2}$ using the global distance-based cross-platform algorithm where the random unitary are sampled from a 4-design, to estimate $\tr{\rho\sigma}$ with $\varepsilon$ error and high success probability.
\label{thm:Global_crossplatform}
\end{theorem}

\begin{proof}
Define the estimator as
\begin{equation}
    v = 2^n \sum_{\bm s, \bm s'} (-2)^{-D(\bm s, \bm s')}\hat{p}_{\bm s}(\Wcal)\hat{q}_{\bm s'} (\Wcal),
    \label{eq:cpEstimator1}
\end{equation}
where $\Wcal$ is sampled from a 4-design. 
By the definition of $\bar{v}$ in Eq.~\eqref{eq:cross_platformEstimator}, we have $\var\pbra{\bar{v}} = \frac{\var\pbra{v}}{N}$. Since the randomness of the estimator $\bar{v}$ comes from both $\Wcal$ and measurement $\bm s$, we denote $v$ as $v(S,\Wcal)$, where $S$ denotes the set of all the measured shots.  Then, due to the total variance formula, we have
\begin{equation}
\var\pbra{v(S,\Wcal)} = \var_{\Wcal}\pbra{\Ebb_S\sbra{v(S,\Wcal)|\Wcal}} + \Ebb_{\Wcal} \sbra{ \var_{S}\pbra{v(S,\Wcal)|\Wcal}}.
\label{eq:totalvariance_sketch}
\end{equation}
We begin by bounding the first part as follows:
\begin{align}
&\var_{\Wcal}\pbra{\Ebb_S\sbra{v(S,\Wcal)|\Wcal}}
=\var_{\Wcal} \pbra{2^n \sum_{\bm s, \bm s'} (-2)^{-D(\bm s, \bm s')}p_{\bm s}(\Wcal) q_{\bm s'} (\Wcal)}\\
&\leq 4^n \Ebb_{W}\sbra{\sum_{\bm s, \bm s'} (-2)^{-D(\bm s, \bm s')} \braket{\bm s|W\rho W^\dagger |\bm s}\braket{\bm s'|W\sigma W^\dagger|\bm s'}\sum_{\bm b, \bm b'} (-2)^{-D(\bm b, \bm b')} \braket{\bm b|W\rho W^\dagger |\bm b}\braket{\bm b'|W\sigma W^\dagger|\bm b'}}\\
&= 4^n\sum_{\bm s, \bm s',\bm b, \bm b'}(-2)^{-D(\bm s, \bm s')-D(\bm b, \bm b')} \Ebb_W\sbra{ \tr{\pbra{\rho\otimes \sigma}^{\otimes 2}W^{\dagger\otimes 4} \ket{\bm s}\bra{\bm s} \otimes \ket{\bm s'}\bra{\bm s'} \otimes \ket{\bm b}\bra{\bm b} \otimes \ket{\bm b'}\bra{\bm b'} W^{\otimes 4} }}
\\
&= d^2 \sum_{\zeta,\tau\in S_4} \sum_{\bm s, \bm s',\bm b, \bm b'}(-2)^{-D(\bm s, \bm s')-D(\bm b, \bm b')} C_{\zeta, \tau}\tr{\pbra{\rho \otimes \sigma}^{\otimes 2} S_{\zeta}}\tr{S_{\tau}\ket{\bm s}\bra{\bm s} \otimes \ket{\bm s'}\bra{\bm s'} \otimes \ket{\bm b}\bra{\bm b} \otimes \ket{\bm b'}\bra{\bm b'} }\label{eq:haar_4thmoment}\\
&\leq c_1 d^2\max_{\tau}\sum_{\sigma\in S_4} C_{\zeta, \tau} \tr{\pbra{\rho \otimes \sigma}^{\otimes 2} S_{\zeta}} 4^n
\label{eq:sum_sb}\\
&\leq c_1 d^4 \max_{\tau}\sum_{\sigma\in S_4} C_{\zeta, \tau}\\
& = c_1 \frac{d^4}{(d+3)(d+2)(d+1)d}
\label{eq:sum_4thmoment}\\
&\leq c_2,
\end{align}
for some constant $c_1,c_2$ where Eq. \eqref{eq:haar_4thmoment} holds by Lemma \ref{lem:Haar_four} and Eq. \eqref{eq:sum_4thmoment} holds by Proposition \ref{pro:Haar_coefs}.

Eq.~\eqref{eq:sum_sb} holds by simplifying the summation of the series~\cite{graham1989concrete}:
\begin{align*}
 & \sum_{\bm b, \bm b',\bm s, \bm s'\in\cbra{0,1}^n} \sum_{\tau\in S_4} (-2)^{-D(\bm s, \bm s') -D(\bm b, \bm b')} \tr{S_{\tau} \ket{\bm s}\bra{\bm s} \otimes \ket{\bm s'}\bra{\bm s'} \otimes \ket{\bm b}\bra{\bm b} \otimes \ket{\bm b'}\bra{\bm b'}}= c_1 4^n 
\end{align*}
for some constant $c_3$. Next, to bound the second term, we note that when a single measurement is performed in both platforms, obtaining bitstrings $\bm s$ and  $\bm s'$ respectively. Then the estimation equals 
\begin{align}
    v_0\pbra{\cbra{\bm s, \bm s'},\Wcal} = 2^n (-2)^{-D(\bm s, \bm s')}.
\end{align}
Therefore, the variance of the estimation equals
\begin{align}
\var_S\pbra{v_0\pbra{\cbra{\bm s, \bm s'}, \Wcal}} \leq \Ebb_{S}\sbra{v_0\pbra{\cbra{\bm s, \bm s'}, \Wcal}^2} = 4^n \sum_{\bm s, \bm s'}(-2)^{-2D(\bm s, \bm s')} p_{\bm s}\pbra{\Wcal} q_{\bm s'}\pbra{\Wcal}.
\end{align}
The variance for the estimation with $m$ measurements equals
\begin{align}
    \var_S\pbra{v\pbra{S, \Wcal}} = \frac{\var_S\pbra{v_0\pbra{S, \Wcal}}}{m^2} \leq \frac{4^n}{m^2} \sum_{\bm s, \bm s'}(-2)^{-2D(\bm s, \bm s')} p_{\bm s}\pbra{\Wcal} q_{\bm s'}\pbra{\Wcal}.
    \label{eq:variance_inner_part}
\end{align}
Hence, the second term of Eq.~\eqref{eq:totalvariance_sketch} can be bounded by
\begin{align}
\Ebb_{\Wcal} \sbra{\var_S\pbra{v(S,\Wcal)|\Wcal}} &\leq 4^n \sum_{\bm s, \bm s'} p_{\bm s}(\Wcal) q_{\bm s'}(\Wcal) (-2)^{-2D(\bm s, \bm s')}\frac{1}{m^2}\\
&=\frac{d^2}{m^2} \frac{c_4}{d^2} 2^n \sum_{k=0}^n \binom{n}{k} 2^{-2k}
\label{eq:haar_second}\\
&=\frac{c_5}{m^2}\pbra{\frac{5}{2}}^n,
\label{eq:varSEU}
\end{align}
for some constant $c_4,c_5$ for any $\rho, \sigma$, where Eq. \eqref{eq:haar_second} holds by Lemma \ref{lem:Haar_second}.   
Hence we can bound the variance to $\var(v(S,\Wcal)) \leq c_2 + \frac{c_5 }{m^2}(5/2)^n$. To ensure the estimator error of $\bar{v}$ remains below $\varepsilon$ with high probability, we require 
\begin{align}
N&\geq c_2/\varepsilon^2,\\
Nm^2&\geq c_5 \pbra{\frac{5}{2}}^n/\varepsilon^2,
\end{align}
which implies
$Nm = \sqrt{N} \times \sqrt{Nm^2}\geq \Ocal\pbra{\sqrt{(5/2)^n}/\varepsilon^2}$.

\end{proof}

Different from the results of Theorem~\ref{thm:Global_crossplatform} rooted on the global 4-design, we next give the analysis for the estimator in Eq.~\eqref{eq:cross_platformEstimator} when the random unitary is picked from a local 4-design.
\begin{theorem}
The number of measurements can be bounded by $\Ocal\pbra{\sqrt{(18/5)^n}/\varepsilon^2}$ using the local distance-based cross-platform algorithm where the random ensemble is chosen as the tensor product of a local 4-design, to solve $\tr{\rho\sigma}$ with $\varepsilon$ error and high success probability if $\rho = \sigma$.
\label{thm:Global_crossplatform}
\end{theorem}
\begin{proof}
Here we also consider the estimator with $N=1$, denoted as  
\[v = 2^n \sum_{\bm s, \bm s'} (-2)^{-D(\bm s, \bm s')}\hat{p}_{\bm s}(\Wcal)\hat{q}_{\bm s'} (\Wcal),\] where $\Wcal$ is sampled from the tensor product of a local 4-design.
With the constraint that $\rho=\sigma$, the first part of Eq.~\eqref{eq:totalvariance_sketch} can be bounded by
\begin{align}
&\var_{\Wcal}\pbra{\Ebb_S\sbra{v(S,\Wcal)|\Wcal}} \\
&= \var_{\Wcal} \pbra{2^n \sum_{\bm s, \bm s'} (-2)^{-D(\bm s, \bm s')}p_{\bm s}(\Wcal) q_{\bm s'} (\Wcal)}
\\
&\leq 4^n \tr{\pbra{\rho \otimes \sigma}^{\otimes 2} 
\prod_{j=1}^n \sum_{\zeta_j,\tau_j\in S_4} 
V_{\zeta_j}
C_{\zeta_j,\tau_j} \sum_{s_j,s_j',b_j,b_j'} (-2)^{-D(s_j,s_j')-D(b_j,b_j')}
\tr{V_{\tau_j} \superket{s_j} \superket{s_j'} \superket{b_j} \superket{b_j'} }
}\\
&= 4^n \prod_{j=1}^n \sum_{\tau_j\in S_4}\sum_{\zeta_j\in S_4} 
 C_{\zeta_j,\tau_j} \sum_{\tau_j}\sum_{s_j,s_j',b_j,b_j'} 2^{-D(s_j,s_j')-D(b_j,b_j')}
\tr{V_{\tau_j} \superket{s_j} \superket{s_j'} \superket{b_j} \superket{b_j'}}\\
& =  4^n \prod_{j=1}^n \frac{36}{5!}\label{eq:local_cross_cut_bound_deviation}\\
&=\pbra{\frac{6}{5}}^n,
\end{align}
where Eq.~\eqref{eq:local_cross_cut_bound_deviation} holds by Proposition~\ref{pro:Haar_coefs} and Lemma \ref{lem:intermediate_result_crossplatform}. We next consider the second part of the variance:
\begin{align}
&\Ebb_{\Wcal} \sbra{ \var_{S}\pbra{v(S,\Wcal)|\Wcal}} \\
&\leq \frac{4^n}{m^2} \sum_{\bm s, \bm s'}  2^{-2D(\bm s, \bm s')}\Ebb_{\Wcal}\sbra{p_{\bm s}(\Wcal) q_{\bm s} (\Wcal)}
\label{eq:variance_partial}\\
&= \frac{4^n }{m^2} \tr{(\rho \otimes \sigma) \prod_{j=1}^n \sum_{\zeta_j,\tau_j\in S_2} C_{\zeta_j, \tau_j}\sum_{s_j,s'_j} (-2)^{-2D(s_j,s'_j)} \tr{V_{\tau_j} \superket{s_j} \superket{s'_j} } }
\label{eq:haar_meas_second}\\
&\leq \frac{4^n}{m^2} \prod_{j=1}^n \max_{\tau_j}\sum_{\zeta_j} C_{\zeta_j, \tau_j} \sum_{\tau_j} \sum_{s_j,s_j'} (-2)^{-2D(s_j,s'_j)} \tr{V_{\tau} \superket{s_j} \superket{s'_j} }\\
&=\frac{4^n}{m^2} \prod_{j=1}^n \pbra{\frac{9/2}{3!}}\label{eq:deviation_second_moment}\\
&=\frac{3^n}{m^2},
\end{align}
where Inequality~\eqref{eq:variance_partial} follows from Inequality~\eqref{eq:variance_inner_part}, and Eq.~\eqref{eq:haar_meas_second} follows from Lemma~\ref{lem:Haar_four}. Eq.~\eqref{eq:deviation_second_moment} holds because $\sum_{\zeta_j} C_{\zeta_j, \tau_j} = \frac{1}{3!}$, as established in Proposition~\ref{pro:Haar_coefs}, and
\begin{equation}
    \sum_{s_j,s'_j} (-2)^{-2D(s_j, s'_j)} \tr{V_{\tau} \superket{s_j} \superket{s'_j}} = \frac{9}{2}
\end{equation}
which is obtained using similar enumeration methods as in the proof of Lemma~\ref{lem:intermediate_result_crossplatform}.

By combining these two inequalities, we have
\begin{align}
  \var\pbra{v(S,\Wcal)}\leq \pbra{\frac{6}{5}}^n + \frac{3^n}{m^2}.
  \label{eq:cross_v_variance}
\end{align}
Hence to ensure the estimator error of $\bar{v} = \sum_{i=1}^N \frac{v_i}{N}$ remains below $\varepsilon$, we require
\begin{align}
N&\geq \pbra{\frac{6}{5}}^n/\varepsilon^2,\\
Nm^2&\geq 3^n/\varepsilon^2,
\end{align}
then we have $Nm\geq \frac{\sqrt{(18/5)^n}}{\varepsilon^2}$.
\end{proof}

Through our analysis, we establish an upper bound on $ Nm$ for any input states $\rho$ and $\sigma$, as shown in the following proposition.
\begin{proposition}
The number of measurements can be bounded by $\Ocal\pbra{\sqrt{6^n}/\varepsilon^2}$ using the local distance-based cross-platform algorithm where the random ensemble is chosen as the tensor product of a local 4-design, to solve $\tr{\rho\sigma}$ with $\varepsilon$ error and high success probability.
\label{pro:cross_platform_general}
\end{proposition}
\begin{proof}
    When $\rho \ne \sigma$, the variance $\var_{\Wcal}\pbra{\Ebb_S\sbra{v(S,\Wcal)|\Wcal}}$ can be expressed as
\begin{align}
\var_{\Wcal}\pbra{\Ebb_S\sbra{v(S,\Wcal)|\Wcal}}
&\leq 4^n \tr{\pbra{\rho \otimes \sigma}^{\otimes 2} 
\prod_{j=1}^n \sum_{\zeta_j,\tau_j\in S_4} 
V_{\zeta_j}
C_{\zeta_j,\tau_j} \sum_{s_j,s_j',b_j,b_j'} (-2)^{-D(s_j,s_j')-D(b_j,b_j')}
\tr{V_{\tau_j} \superket{s_j} \superket{s_j'} \superket{b_j} \superket{b_j'} }
}\\
&=\sum_{\bm \sigma, \bm \tau\in S_4^{n}} C_{\bm \sigma, \bm \tau}\tr{\pbra{\rho \otimes \sigma}^{\otimes 2} V_{\bm \sigma}} \sum_{\bm s, \bm s', \bm b, \bm b'\in\cbra{0,1}^n} (-2)^{-D(\bm s, \bm s')-D(\bm b,\bm b')}\tr{V_{\bm \tau} \superket{\bm s} \superket{\bm s'} \superket{\bm b} \superket{\bm b'}}\\
&\leq 4^n \prod_{j=1}^n \max_{\tau_k}\sum_{\zeta_j\in S_4} 
 C_{\zeta_j,\tau_j} A\\
 & = 4^n \prod_{j=1}^n \max_{\tau_j}\sum_{\zeta_j\in S_4} 
 C_{\zeta_j,\tau_j} \times 60\\
 &= 2^n,
\end{align}
where $C_{\bm \sigma, \bm \tau} = \prod_j C_{\sigma_j,\tau_j}, V_{\bm \sigma} = \otimes_j V_{\sigma_j}$, where $A$ contains the terms of $(-2)^{-D(s_j,s_j')-D(b_j,b_j')}
\tr{V_{\tau_j} \superket{s_j} \superket{s_j'} \superket{b_j} \superket{b_j'}}$ such that $(-2)^{-D(s_j,s_j')-D(b_j,b_j')}>0$, which equals $60$ by counting on the associated terms in Supplementary Table \ref{supptab:intermediate_proof_tab}.
It implies 
\begin{align}
\var\pbra{v(S,\Wcal)}\leq {2^n + \frac{3^n}{m^2}}.
\label{eq:variance_cross_platform_general}
\end{align}
Hence $Nm\geq \sqrt{6^n}/\varepsilon^2$.
\end{proof}
 
\noindent\textbf{Remark}. Based on the above theorems and proposition, we observe that when global distance-based cross-platform protocols are permitted, the number of required measurements scales as $\Ocal\{\sqrt{(5/2)^n}/\varepsilon^2\}$—worse than the collision-based cross-platform approach, which requires only $\Ocal\{\sqrt{2^n}/\varepsilon\}$ measurements. If restricted to local distance-based protocols and assuming $\rho = \sigma$, the measurement cost scales as $\Ocal\{\sqrt{(18/5)^n}/\varepsilon\}$. In the general case with arbitrary $\rho$ and $\sigma$, the scaling becomes $\Ocal\{\sqrt{6^n}/\varepsilon^2\}$. Nevertheless, this remains an improvement over full quantum state tomography, which requires $\Ocal(4^n)$ samples~\cite{Gross2010Quantum, riofrio2017experimental}.

\section{Algorithmic details}
\label{supp:algorithm_details}

In this section, we detail the algorithmic implementation of the proposed method for estimating cross-platform fidelity using wire-cutting techniques. To illustrate the approach, we start with the simplified case of a single cut, implementing wire cutting in parallel. This forms the core algorithm employed in our experiments.
We then present the proof for an extreme multicut case in SI~\ref{SI:IV-a}, and extend the analysis to accommodate more general circuit structures in SI~\ref{SI:IV-B}.

\subsection{Implementation of wire cutting in parallel}
\label{subsec:in_parallel_method}
Recall that the general procedure of wire cutting proceeds sequentially~\cite{Lowe2023fast}, which results in high runtime costs in practice, particularly in the proposed cross-platform fidelity estimation algorithm. To address this bottleneck, we develop an improved version of the proposed algorithm. Specifically, in the case of $k_1=1$, the wire cutting in each platform can be executed in parallel, leading to a substantial reduction in runtime. 

The improved cutting process of one platform, as visualized in Supplementary Figure~\ref{method_fig:compress_innerproduct}, contains two modifications. The first modification simplifies the sampling of the random cutting operation by replacing the original 24-element Clifford group with a three-element unitary set. Mathematically, the set of possible operations $\left( z^{(l)}, U_{z^{(l)}}, p_{z^{(l)}} \right)$ in Eq.~(\ref{eqn:wire-cut}) is given by
\[
\mathcal{Z} = \cbra{(0, \mathbb{I}, \tfrac{1}{5}), (0, R_y^{\frac{\pi}{2}}, \tfrac{1}{5}), (0, R_x^{\frac{\pi}{2}}, \tfrac{1}{5}), (1, \mathbb{I}, \tfrac{2}{5})},
\]
where $z^{(l)} \in \{0,1\}$ represents the outcome sampled from a Bernoulli distribution; $U_0 \in \cbra{\mathbb{I}, R_y^{\frac{\pi}{2}}, R_x^{\frac{\pi}{2}}}$, and $U_1 = \mathbb{I}$. Here, $R_\sigma^\theta := \exp(-i \theta \sigma / 2)$ denotes a rotation about the axis $\sigma \in \cbra{X, Y}$ by angle $\theta$, with $\sigma$ the corresponding Pauli operator. The quantity $p_{z^{(l)}}$ specifies the probability associated with each tuple. This simplification sets the stage for the second modification.

Before presenting the second modification, let us review the original procedure of wire cutting with $k_1=k_2=1$. That is, when a given circuit is cut by a single qubit, it divides into two subcircuits, \textsf{Q}-\textsf{A} and \textsf{Q}-\textsf{B}. The circuit configuration of \textsf{Q}-\textsf{B} depends on the sampled $( {z}^{(l)}, U_{{z}^{(l)}} )$ applied to \textsf{Q}-\textsf{A} and the corresponding measurement outcome, leading to the sequential nature.  

To avoid the sequential execution, the second modification involves independently evaluating all possible circuit configurations for \textsf{Q}-\textsf{A} and \textsf{Q}-\textsf{B}, followed by classical post-processing. In particular, thanks to the reformulated operation set $\mathcal{Z}$, there are $4$ circuit configurations for \textsf{Q}-\textsf{A} and $2\times 4$ configurations for \textsf{Q}-\textsf{B}, where the factor $2$ stems from the initialized state for the cutting qubit being $\ket{0}$ or $\ket{1}$. For each circuit configuration in \textsf{Q}-\textsf{A} or \textsf{Q}-\textsf{B}, the random unitary $\Wcal^{\bar{A}}$ or $\Wcal^{B}$ with $\Wcal=\Wcal^{\bar{A}} \otimes \Wcal^{B}$ sampled from $\Clgroup_1^{\otimes n}$ is applied before the measurement, where $\bar A$ refers to the set of all qubits except the cutting qubit in \textsf{Q}-\textsf{A}. Define the generated random unitaries as  $\{\Wcal_t=  \Wcal_t^{\bar{A}} \otimes \Wcal_t^{B}\}_{t=1}^N$. As such, the total number of circuit configurations is $(4+8)\times N=12N$, which can be executed in parallel. For each specific configuration, we perform $m$ measurement shots to obtain an estimation.

The classical post-processing aims to leverage the measurement outcomes from $12N$ circuit configurations from two platforms to estimate $\tr{\rho\sigma}$. 
For clarity, let $c$ and $c'$ denote the measurement outcomes on the cutting qubit of \textsf{Q}-\textsf{A} from Platform~1 and Platform~2, respectively. The corresponding input to \textsf{Q}-\textsf{B} on Platform~1 is set to $c$ if $z^{(l)} = 0$, and to $\mathbb{I}/2$ if $z^{(l)} = 1$. The same rule applies for Platform~2.
Moreover, denote the $i$-th tuple of $\mathcal{Z}$ as $\mathcal{Z}_i$, and $\mathcal{Z}_{i,j}$ as the $j$-th element of $\mathcal{Z}_i$. The estimator of $\tr{\rho\sigma}$ takes the form as 
\begin{widetext}
  \begin{align}
\begin{aligned}
    \hat{v}&:=\frac{25}{ N^2} 2^n\sum_{j=1}^{4}\sum_{j' = 1}^4 (-1)^{\Zcal_{j,1} + \Zcal_{j',1}} \Zcal_{j,3} \Zcal_{j',3}  \sum_{c,c'\in \cbra{0,1}} \sum_{\bm a, \bm a'}\sum_{t=1}^{N} (-2)^{-D(\bm a,\bm a')}\hat{p}(\bm a,c|\Zcal_j, \Wcal_t^{\bar A}) \hat{q}(\bm a',c'|\Zcal_{j'}, \Wcal_t^{\bar A})\\
& \sum_{\bm{b}, \bm{b'}}\sum_{t=1}^{N} (-2)^{-D(\bm{b}, \bm{b'})}\hat{p}(\bm b|\Zcal_j, c,\Wcal_t^{B}) \hat{q}(\bm b'|\Zcal_{j'}, c'_{z},\Wcal_t^{B}),
\end{aligned}
\label{eq:estimator_parallel_cutting}
\end{align}  
\end{widetext}
where $\hat{p}(\cdot)$ and $\hat{q}(\cdot)$ are the associated estimated probabilities with a fixed circuit configuration and a specific input-and-output pair of the cutting qubit. The explicit formula is summarized below:
\begin{itemize}
\item $ \hat{p}\pbra{\bm a,c|\Zcal_j, \Wcal_t} :=\frac{\# (\bm a,c)}{m}$, where $  t \in [N]$, $  j \in [4]$, and $  (\bm a, c) $  are the sampled bit-strings in \textsf{Q}-\textsf{A} of Platform 1.
\item $ \hat{q}\pbra{\bm a',c'|\Zcal_j, \Wcal_t} :=\frac{\# (\bm a', c')}{m}$, where $  t \in [N]$, $  j \in [4]$, and $  (\bm a', c') $  are the sampled bit-strings in \textsf{Q}-\textsf{A} of Platform 2.
\item $ \hat{p}\pbra{\bm b |\Zcal_j,c,\Wcal_t} :=\frac{\#\bm b | s_q = c}{m}$, where $  t \in [N]$, $  j \in [3]$, $s_q$ is the input of the cutting qubit in \textsf{Q}-\textsf{B} of Platform 1, and $ \bm b $  is the bit-string in \textsf{Q}-\textsf{B} of Platform 1.
\item $ \hat{q}\pbra{\bm b'  |\Zcal_j,c',\Wcal_t} :=\frac{\# \bm b'  | s'_q = c'}{m}$, where $  t \in [N]$, $  j \in [3]$, $  s'_q$ is the input of the cutting qubit in \textsf{Q}-\textsf{B} of Platform 2, and $  b'   $  is the bit-string in \textsf{Q}-\textsf{B} of Platform 2.
\item $ \hat{p}\pbra{\bm b |\Zcal_j,\frac{\Ibb}{2}, \Wcal_t} :=\frac{1}{2} \frac{\sum_{i\in\cbra{0,1}}\#\bm b|s_q = i}{m}$, 
where $  t \in [N]$, $  j = 4$, and $  s_q $  is the input of the cutting qubit in \textsf{Q}-\textsf{B} of Platform 1. This corresponds to the case where the input state of the cutting qubit of \textsf{Q}-\textsf{B} is the maximally mixed state $\frac{\Ibb}{2}$. With slight abuse of notation, we continue to use $c$ to refer to this revised input $c = \frac{\Ibb}{2}$.
\item $ \hat{q}\pbra{\bm b'  |\Zcal_j,\frac{\Ibb}{2}, \Wcal_t} :=\frac{1}{2} \frac{\sum_{i\in\cbra{0,1}}\# \bm b'  |s'_q = i}{m}$, where $t \in [N]$, $j = 4$, and $  s'_q $  is the input of the cutting qubit in \textsf{Q}-\textsf{B} of Platform 2.
\end{itemize}
\begin{figure}
    \centering
\includegraphics[width=1.0\linewidth]{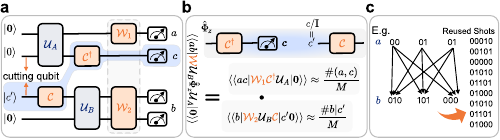}
    \caption{Illustration of parallel sampling and classical post-processing. \textbf{a.} The original circuit with a cut, where the bottom subcircuit depends on the output of the top subcircuit. \textbf{b.} The sequential and parallel sampling processes are shown to be equivalent. The upper and bottom parts can be paralleled by enumerating all possible cutting variables $\pbra{z, \Ccal,p_{z}}$.
    \textbf{c.} By sampling the top and bottom subcircuits independently, each over all possible configurations—$\mathcal{O}(N^2)$ matched circuit pairs can be constructed using only $\mathcal{O}(N)$ samples from each subcircuit.}
\label{method_fig:compress_innerproduct}
\end{figure}

\begin{lemma}
    $\hat{v}$ defined in Eq. \eqref{eq:estimator_parallel_cutting} is an unbiased estimation of $\tr{\rho \sigma}$.
\end{lemma}
\begin{proof}
We employ the Liouville (or vectorized) representation of linear operators and quantum channels. In this representation, a linear operator $A$ is mapped to a normalized vector $\superket{A} := A / \sqrt{\tr{A A^\dagger}}$. The dual vector $\superbra{B}$ corresponds to the Hilbert-Schmidt inner product, such that $\superbraket{B | A} = \tr{B^\dagger A} / \sqrt{\tr{A A^\dagger} \tr{B B^\dagger}}$. Additional details on the Liouville representation are provided in the SI~A.

In the following, we prove that $\Ebb\sbra{\hat{v}} = \tr{\rho \sigma}$. To be concrete, we define $\Phi_{\Zcal_j}=\Zcal_{j,2}^\dagger \sum_{\bm s} \superket{\bm s} \superbra{\bm s} \Zcal_{j,2}$ when $j\in[3] $ and denote $\Phi_{\Zcal_4} = \superket{\Ibb}\superbra{\Ibb}$ which is associated with the fixed circuit configuration $\Zcal_j$.
Note that if $j\in [3]$, we have
\begin{align}
&\Ebb_{\hat{p}}\sbra{\sum_{c} \hat{p}\pbra{\bm a}\hat{p}\pbra{\bm b |c}}\\
&= \sum_{c}\superbra{\bm a, c}\Wcal^{\bar A}_t\Zcal_{j,2} \Ucal_A \superket{\bm 0} \superbra{\bm b } \Wcal^{B}_{t'} \Ucal_B\Zcal_{j,2}^\dagger \superket{c, \bm 0}\\
&= \superbra{\bm a,\bm b}\Wcal_u \Ucal_B \Zcal_{j,2}^{\dagger}\sum_{c}\superket{c}\superbra{c} \Zcal_{j,2} \Ucal_A \superket{\bm 0}\\
&=\superbra{\bm a,\bm b}\Wcal_u \Ucal_B\Phi_{\Zcal_j}\Ucal_A \superket{\bm 0}
\end{align}
where $\Wcal_u = (\Wcal_t^{A\backslash c},\Wcal_{t'}^B)$, and $u\in [N^2]$, $\Ebb\pbra{\hat p}$ denotes the expectation over the measurement condition on fixed $\Zcal_j$ and $\Wcal_t$, and here we omit the input $\Zcal_j, \Wcal_t$ in $\hat{p}(\cdot)$ for convenience. 
We also illustrate this process in Supplementary Figure~\ref{method_fig:compress_innerproduct}\textbf{b,c}.
If $j=4$, we note $c=\frac{\Ibb}{2}$,
\begin{align}
&\Ebb_{\hat{p}}\sbra{\sum_{c} \hat{p}\pbra{a}\hat{p}\pbra{\bm b |c}} \\
&= \sum_{c, s_q} \frac{1}{2}\superbra{\bm s_{\bar A} c}\Wcal^{\bar A}_t\Zcal_{j,2} \Ucal_A \superket{\bm 0} \superbra{\bm b } \Wcal^{B}_{t'} \Ucal_B\Zcal_{j,2}^\dagger \superket{s_q \bm 0}\\
&= \superbra{s_{\bar c}\bm b}\Wcal_u \Ucal_B \Zcal_{j,2}^{\dagger}\frac{1}{2}\sum_{c, s_q}\superket{s_q}\superbra{c} \Zcal_{j,2} \Ucal_A \superket{\bm 0}\\
&=\superbra{\bm a,\bm b}\Wcal_u \Ucal_B \Zcal_{j,2}^{\dagger}\superket{\Ibb} \superbra{\Ibb} \Zcal_{j,2} \Ucal_A \superket{\bm 0}\\
&=\superbra{\bm a,\bm b}\Wcal_u \Ucal_B\Phi_{\Zcal_j}\Ucal_A \superket{\bm 0}.
\end{align}

By substituting into Eq. \eqref{eq:estimator_parallel_cutting}, we have
\begin{widetext}
  \begin{align}
\Ebb\sbra{\hat{v}}&=   \Ebb\sbra{\frac{25}{N^2}  2^n \sum_{j,j'\in[4]} (-1)^{\Zcal_{j,1} + \Zcal_{j',1}} \Zcal_{j,3} \Zcal_{j',3} \sum_{u=1}^{N^2}  \sum_{\bm a,\bm a',\bm b,\bm b'} (-2)^{-D(a, a')-D(\bm b, \bm b')} \superbra{\bm a\bm b}\Wcal_u \Ucal_B\Phi_{\Zcal_j}\Ucal_A \superket{\bm 0}  \superbra{\bm a'\bm b' }\Wcal_u \Vcal_B\Phi_{\Zcal_{j'}}\Vcal_A \superket{\bm 0} }\\
 &=\Ebb\sbra{\sum_{j,j'} \Zcal_{j,3} \Zcal_{j',3}\tr{\rho_{\Zcal_j} \sigma_{\Zcal_{j'}}} }
 \label{eq:cross_platform}
 \\
 &=\tr{\rho \sigma}.
 \label{eq:wire_cut}
\end{align}
\end{widetext}
where Eq.~\eqref{eq:cross_platform} follows from the cross-platform fidelity framework~\cite{Elben2020cross, brydges2019probing}, while Eq.~\eqref{eq:wire_cut} is derived using the wire cutting technique~\cite{Lowe2023fast}.

\end{proof}

\subsection{Cross-platform fidelity estimation with regular cutting}\label{SI:IV-a}

The state preparation circuit in the extreme case takes the form of a staircase composed of sequential sub-blocks, as illustrated in Supplementary Figure~\ref{suppfig:cuts_calculation}\textbf{a}. We consider $k_2$ rounds of cutting, resulting in $k_2 + 1$ parts after all the cuts. For simplicity, we assume that exactly $k_1$ qubits are cut in each round. However, this can easily be generalized to the cases where up to $k_1$ qubits are cut.

In this scenario, two types of cutting qubits are introduced at the $j$-th cut ($1 \leq j \leq k_2$):
\begin{itemize}
    \item[(1)] \textit{incoming cut} — qubits that are cut and measured within \textsf{Q}-$j$. We denote $\bm c_j\in\cbra{0,1}^{k_1}$ be the measurement outcomes in this \textit{incoming cut}.
    \item[(2)] \textit{departing cut} — qubits that were cut and measured in \textsf{Q}-$j$ and are re-initialized in the current \textsf{Q}-$(j+1)$. We denote the initialization state as $\bm{c}_{j-1,\bm{z}}$, which depends on the measurement outcomes of the \textit{incoming cut} in \textsf{Q}-$j$ and the Bernoulli variable associated with this cut, $z_j$, where $\bm{z} = \pbra{z_1, \ldots, z_{k_2}}$.
\end{itemize}

For each cut, we need to apply a random Clifford unitary, as illustrated in Supplementary Figure~\ref{suppfig:cuts_calculation}\textbf{b}. At the $j$-th cut, the Clifford gate operating on the \textit{departing cut} in \textsf{Q}-$(j+1)$ after re-initialization equals the inverse of the Clifford gate operating on the \textit{incoming cut} in \textsf{Q}-$j$ before measurement.

For notational convenience and to emphasize that the two platforms undergo the same process but with independently sampled variables, we use a variable and its primed counterpart to denote corresponding values in Platform~1 and Platform~2, respectively. For example, we use $\bm{z}$ and $\bm{z}'$ to represent the associated Bernoulli vectors.
We give a brief explanation of the algorithm and the generation of variables and random circuits involved in the process as follows.

\begin{itemize}
\item [(1)] \textbf{Generation of cutting-related variables and circuits:}
Let $\bm z = (z_1, \ldots, z_{k_2})$ and $\bm z' = (z'_1, \ldots, z'_{k_2})$.
For the $j$-th cut, we generate independent Bernoulli random variables
$z_j, z_j'\in\cbra{0,1}$ for Platform 1 and Platform 2 that take the value $1$ with probability $2^{k_1} / \pbra{2^{k_1+1} + 1}$. According to  Eqs.~\eqref{eq:phi_0} and \eqref{eq:phi_1}, if $z_j = 0$ (or $z_j'=1$), the random Clifford gate $\Ccal_j$ (or $\Ccal'_j$) from $\Cl_{k_1}$ is applied; otherwise, the applied gate is $\Ccal_j = \Ibb$ (or $\Ccal'_j= \Ibb$ for Platform 2). Recall that if $z_j=0$, then a random Clifford gate is performed, and if $z_j=1$ then the identity gate is performed.
\item [(2)]  \textbf{Generation of fidelity-related local random unitaries:}
We define the set of random unitary channels as $\Wmat = \cbra{\Wcal^{(1)}, \ldots, \Wcal^{(N)}}$, where each $\Wcal \in \Wmat$ is a tensor product of local unitaries. Here, with slight abuse of notation, $\Wcal_j$ denotes the random unitary applied to \textsf{Q}-$j$, i.e., $\Wcal = \bigotimes_{j=1}^{k_2+1} \Wcal_j$, and $\Wmat_j$ represents the associated ensemble, as illustrated in Supplementary Figure~\ref{suppfig:cuts_calculation}(b). 
\item [(3)]  \textbf{Specification of outcome probabilities:}  
Recall that $\bm{c}_j$ denotes the measurement outcomes of the \textit{incoming cut} in \textsf{Q}-$j$. Let $\bm{s}_j$ represent the measurement outcomes of the remaining qubits in \textsf{Q}-$j$, excluding the \textit{incoming cut}.
Define $p(\bm s_j,\bm c_j\mid \bm c_{j-1,z_{j-1}},\Wcal_j)$ as the probability of obtaining outcome $(\bm s_j, \bm c_j)$ with the restriction that the input for the \textit{departing cut} $\bm c_{j-1,z_{j-1}} = \bm c_{j-1}$ if $z_{j-1} = 0$, while $\bm c_{j-1,z_{j-1}}$ represents the maximally mixed state $\Ibb/2^{k_1}$ otherwise.  In the latter case, the probability $p({\bm s}_j, \bm{c}_j \mid \bm{c}_{j-1,\bm{z}}, \mathcal{W}_j)$ is estimated by averaging over the frequencies of $({\bm s}_j,\bm c_j)$, based on repeated sampling where inputs to the \textit{departing cut} in \textsf{Q}-$j$ are randomly drawn from $\cbra{0,1}^{k_1}$. Similarly we define $q(\bm s'_j,\bm c'_j|\bm c'_{j-1,z_{j-1}},\Wcal_j)$ to be the corresponding probability associated with quantum state $\sigma$ prepared in Platform $2$.
\end{itemize}

 In the following, we describe the general circuit structure for \textsf{Q}-$j$ where $1\leq j\leq k_2 + 1$. The $1$-st and $k_2 + 1$-th (last) parts involve only a single cut, whereas all others involve two cuts. Hence, for these two special cases, the operations associated with the additional cut are trivial. 
As illustrated in Supplementary Figure~\ref{suppfig:cuts_calculation}\textbf{b}, the measurement outcomes $\bm{s}_j$ are obtained by applying a random unitary $\Wcal\in \Wmat$ followed by computational basis measurement. In contrast, the measurement outcomes $\bm{c}_j$ correspond to the \textit{incoming cut}, where a Clifford gate is performed with probability associated with the Bernoulli variable $z_j$ before measurement.

 The other input qubits for \textsf{Q}-$j$ are initialized to be $\ket{\bm 0}$, after performing $\Ccal_{j-1}$ on the \textit{departing cut}, followed by operating unitary $U_j$, then a random unitary $\Wcal_j\otimes \Ccal_j^\dagger$ prior to the $Z$-basis measurement in \textsf{Q}-$j$.

The estimators $\superket{\hat{\rho}_{\bm{z}}}$ and $\superket{\hat{\sigma}_{\bm{z}'}}$ extend the definition in Eq.~\eqref{eq:def_rho_z_approx} to the case of $k_2$ cuttings. They represent the estimations of $\rho$ and $\sigma$, respectively, obtained by applying $\Phi_{z}$ channels (with $z \in \cbra{0,1}$) to the $k_2$ cutting qubits.
We give the estimation for $\tr{\hat{\rho}_{\bm z} \hat{\sigma}_{\bm z'}}$ with random unitary set $\Wmat$ as follows,
\begin{align}
\hat v_{\bm z,\bm z'} &:=\frac{2^n(-1)^{\abs{\bm z} + \abs{\bm z'}}}{N^{k_2+1}}\pbra{2^{k_1+1} + 1}^{2k_2} \sum_{\bm c,\bm c'}\prod_{j = 1}^{k_2 + 1}\sum_{
\Wcal_j,\bm s_j,\bm s'_j}
(-2)^{-D(\bm s_j,\bm s'_j)}\hat p(\bm s_j,\bm c_j|\bm c_{j-1,z_{j-1}},\Wcal_j)\hat q(\bm 
 s_j,\bm c'_j|\bm c'_{j-1,z_{j-1}},\Wcal_j),
\label{eq:cross_cutting_general_estimator}
\end{align}
where $\hat{p}(\bm{s}_j, \bm{c}_j \mid \bm{c}_{j-1,\bm{z}}, \Wcal_j)$ and $\hat{q}(\cdot)$ denote the estimated probabilities of $p(\bm{s}_j, \bm{c}_j \mid \bm{c}_{j-1,\bm{z}}, \Wcal_j)$ and $q(\cdot)$, respectively. These estimates correspond to the observed frequencies of the outcomes $\pbra{\bm{s}_j, \bm{c}_j}$ in $m$ rounds of measurements, under the constraints of $\bm{c}_{j-1,\bm{z}}$ and $\Wcal_j$. For convenience, we omit explicit mention of $m$ in the notation.

Assume $L$ Bernoulli variable vectors are generated randomly and uniformly.
Let $\Zmat_1 = (\bm z_1, \bm z_2, \ldots, \bm z_{L})$ (or $\Zmat_2 = (\bm z_1', \bm z_2', \ldots, \bm z_{L}')$) be sets of randomly generated Bernoulli variable vectors, where $\bm z_{i}\in \cbra{0,1}^{k_2}$ in Platform 1 (or Platform 2).
Let $\Zmat = (\Zmat_1, \Zmat_2)$.
For each pair $\bm z \in \Zmat_1$ and $\bm z' \in \Zmat_2$, generate $\hat v_{\bm z,\bm z'}$ using Eq.~\eqref{eq:cross_cutting_general_estimator} with randomly generated unitaries $\Ccal$ and $\Ccal'$ corresponding to $\bm z$ and $\bm z'$, respectively. The final estimator is then given by 
\begin{align}
  \hat v = \sum_{\bm z,\bm z'} \frac{\hat v_{\bm z,\bm z'}}{L^2}.
\label{eq:cross_cutting_Lrounds_general_estimator}
\end{align}

\begin{figure*}
    \includegraphics[width = 1.0\textwidth]{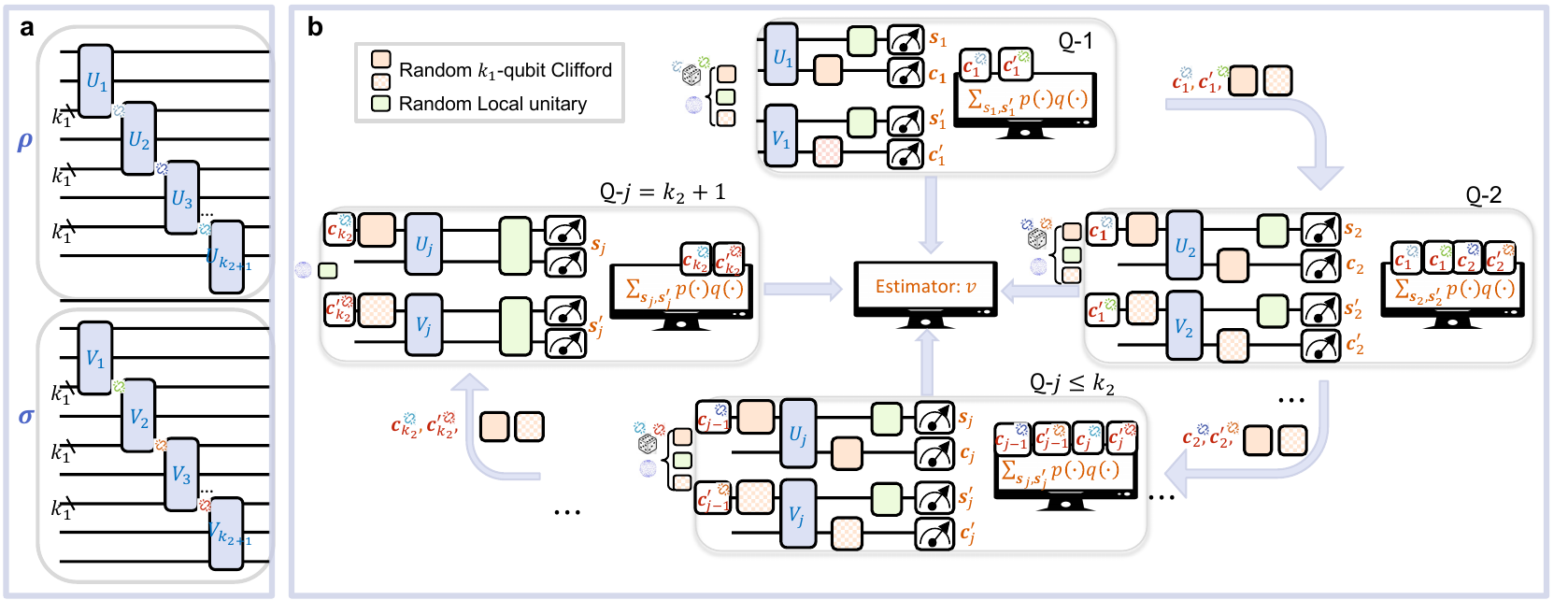}
    \caption{Illustration of constructing the estimator using the cross-platform circuit-cutting algorithm.
\textbf{a.} Original circuits used to prepare the quantum states $\rho$ and $\sigma$, with $k_2$ circuit cuts (represented by cut icons above the cut qubits) indicated on the corresponding wires.
\textbf{b.} Estimation of $\tr{\rho\sigma}$ using the cross-platform fidelity with classical communication algorithm. Each wire in \textsf{Q}-$j$ represents multiple quantum wires. Orange-colored boxes indicate random Clifford gates applied to the cutting qubits, while green-colored boxes denote random unitaries sampled from the tensor product of local random unitaries.}
\label{suppfig:cuts_calculation}
\end{figure*}

\begin{theorem}
Let $\hat{v}$ be the estimator of $\tr{\rho \sigma}$ defined in Eq.~\eqref{eq:cross_cutting_Lrounds_general_estimator}, 
where $\rho$ and $\sigma$ are two $n$-qubit quantum states. Suppose each cut involves at most $k_1$ qubits, and there are $k_2$ total cuts dividing the circuit into $k_2+1$ parts.
Then, $\hat{v}$ is an unbiased estimator of $\tr{\rho \sigma}$, and satisfies $\abs{\hat{v} - \tr{\rho \sigma}} \leq \varepsilon$ with high probability, provided the number of measurements is upper bounded by
\[ \Ocal\pbra{\pbra{2^{k_1+1} + 1}^{2k_2}\frac{6^{n/(2(k_2+1))}}{\varepsilon}}.\]
\label{thm:cut_fidelity_tight}   
\end{theorem}

\begin{proof}
Recall that in the proposed algorithm, $L$ Bernoulli variable vectors are generated randomly and uniformly for both platforms. Accordingly, let $\mathbf{\hat\Psi}_1 = \pbra{\hat{\Psi}^{(1)}, \hat{\Psi}^{(2)}, \ldots, \hat{\Psi}^{(L)}}$, and $\mathbf{\hat\Psi}_2 = \pbra{\hat{\Psi}^{(L+1)}, \hat{\Psi}^{(2)}, \ldots, \hat{\Psi}^{(2L)}}$, be the sampled random channel set associated with quantum states $\rho$ and $\sigma$ prepared in Platform 1 and 2 respectively, where $\hat{\Psi}^{(j)} = (\hat\Psi_1^{(j)},\ldots, \hat\Psi_{k_2}^{(j)})$ for $1\leq j \leq 2L$, and $\hat\Psi_i^{(j)}\in\cbra{\Ccal_i M_{\bm s}\Ccal_i^\dagger, \Ibb}$ is an unbiased estimation of the cutting channels $\Psi_i^{(j)} \in \cbra{\Phi_0,\Phi_1}$ where $\Phi_0 = \Ebb_{\Ccal}\sbra{\Ccal\circ M_{\bm s} \circ \Ccal^{\dagger}}$ and $\Phi_1 = \superket{\Ibb}\superbra{\Ibb}$ as defined in Eq.~\eqref{eq:phi_0} and \eqref{eq:phi_1}, $M_{\bm s} = \sum_{\bm s}\superket{\bm s}\superbra{\bm s}$ is the measurement and state preparation channel. In what follows, we denote  $\mathbf{\hat\Psi} = \pbra{\mathbf{\hat\Psi}_1, \mathbf{\hat\Psi}_2}$.

Recall that $\Ebb[\hat v] = \Ebb_{\bm{\hat{\Psi}}, \bm s,\bm s',\Wcal, \bm z, \bm z'}\sbra{\hat v_{\bm z,\bm z'}}$. As such, we split the proof into two parts. Specially, the first part amounts to showing that $\Ebb_{\bm z, \bm z
'}[\tr{\rho_{\bm z}\sigma_{\bm z'}}] = \tr{\rho \sigma}$. This can be easily proved by making use of the fact that $\Ebb\sbra{\rho_{\bm z}} = \rho$ and $\Ebb\sbra{\sigma_{\bm z'}} = \sigma$.

The second part is to prove that $\Ebb_{\bm{\hat{\Psi}},\bm s,\bm s',\Wcal}\sbra{\hat v} = \tr{\rho_{\bm z}\sigma_{\bm z'}}$. Here $\rho_{\bm z}$ and $ \sigma_{\bm z'}$ are the cut states of Platform 1 and Platform 2 after $k_2$ cuttings. It is easy to see that $\Ebb_{\bm{\hat{\Psi}},\bm s,\bm s',\Wcal}\sbra{\hat v} = \Ebb_{\bm{\hat{\Psi}}, \bm s,\bm s',\Wcal}\sbra{\hat v_{\bm z,\bm z'}} $. Hence we only need to prove $\Ebb_{\bm{\hat{\Psi}}, \bm s,\bm s',\Wcal}\sbra{\hat v_{\bm z,\bm z'}} = \tr{\rho_{\bm z}\sigma_{\bm z'}}$.

We start by analyzing the probability of a basis $\bm s$ associated with the state $\rho_{\bm z}$, i.e.,
 \begin{align}
\Ebb_{\hat{\Psi}\in \bm{\hat{\Psi}} }&\sbra{(-1)^{\abs{\bm z}}\pbra{2^{k_1+1} + 1}^{k_2}\sum_{\bm c}\prod_{j = 1}^{k_2 + 1}p\pbra{\bm s^j,\bm c_{j}|\bm c_{j-1,z_{j-1}},\Wcal_j}}\\
 &=\Ebb_{\hat{\Psi}\in \bm{\hat{\Psi}}}\sbra{(-1)^{\abs{\bm z}}\pbra{2^{k_1+1} + 1}^{k_2}\sum_{\bm c}\prod_{j = 1}^{k_2 + 1}\superbraket{\bm s^j,\bm c_{j}|\Wcal_j\Ccal_j\Ucal_j\Ccal_{j-1}|\bm c_{j-1,z_{j-1}},\bm 0} }\\
 &=(-1)^{\abs{\bm z}}\pbra{2^{k_1+1} + 1}^{k_2}\sum_{\bm c}\superbra{\bm s}\Wcal\Ucal_{k_2+1} \Ebb_{\hat{\Psi}_{k_2}}\sbra{\Ccal_{k_2}^{\dagger}\superket{\bm c_{k_2,z}}\superbra{\bm c_{k_2}}\Ccal_{k_2}}\cdots \Ucal_{j+1}\Ebb_{\hat{\Psi}_{j}}\sbra{\Ccal_{j}^\dagger\superket{\bm c_{j,z}}\superbra{\bm c_{j}}\Ccal_{j}}\cdots \Ucal_1 \superket{0^{n}}\\
 &= (-1)^{\abs{\bm z}}\pbra{2^{k_1+1} + 1}^{k_2}\superbraket{\bm s|\Wcal\Ucal_{k_2+1}\Psi_{k_2+1}\Ucal_{k_2}\cdots \Psi_1\Ucal_1|0^{n}}\\
 &:=p\pbra{\bm s\mid\rho_{\bm z}, \Wcal},
\end{align} 
where $\hat{\Psi}= \pbra{\Psi_1,\ldots, \Psi_{k_2}}$, $\bm s = \pbra{\bm s^1, \ldots, \bm s^{k_2+1}}$, and $\Wcal = \otimes_{j = 1}^{k_2+1}\Wcal_j$.  The quantity $p\pbra{{\bm s} \mid \rho_{\bm z}, \Wcal}$ denotes the probability of obtaining outcome ${\bm s}$ when the quantum state $\rho_{\bm{z}}$ is evolved under the channel $\Wcal$ prior to measurement. Similarly, we have
\begin{align}
\begin{aligned}
  &q(\bm s'|\Wcal,\sigma_{\bm z'})= \pbra{2^{k_1+1} + 1}^{k_2}\sum_{\bm c'}\Ebb_{\hat{\Psi}\in \bm{\hat{\Psi}}}\sbra{(-1)^{\abs{\bm z'}}\prod_{j = 1}^{k_2 + 1}q\pbra{\bm s'_j,\bm c'_j|\bm c'_{j-1,z_{j-1}},\Wcal_j}}.  
\end{aligned}
\end{align}
Hence we have
\begin{align}
\begin{aligned}
\Ebb_{\bm s,\bm s',\hat\Psi,\Wcal}\sbra{v_{\bm z,\bm z'}} &= \frac{2^n}{N^{k_2+1}}\Ebb_{\Wcal}\sbra{\sum_{\Wcal}\sum_{\bm s,\bm s'}(-2)^{-D(\bm s,\bm s')} p(\bm s|{\rho}_{\bm z}, \Wcal) q(\bm s'|{\sigma}_{\bm z'}, \Wcal)}\\
    &= \tr{{\rho}_{\bm z} {\sigma}_{\bm z'}},   
\end{aligned}
\label{eq:expect_s_phi_w_v}
\end{align}
where the second equality of Eq.~\eqref{eq:expect_s_phi_w_v} holds by Eq.~\eqref{eq:cross_platFid}.

\smallskip
In the following, we give the bound to the variance.
  Let $\Xmat$ be the measurement set containing all the measurement shots of $\rho$ and $\sigma$. 
  Since $\var\pbra{v} = \frac{\var\pbra{v_{\bm z,\bm z'}}}{L^2}$, it suffices to determine the variance bound of $v_{\bm z,\bm z'}$.

  According to the law of total variance, we express the variance of $v_{\bm z,\bm z'}(\Zmat, \mathbf{\hat{\Psi}}, \Wmat, \Xmat)$ as
    \begin{align}
    \begin{aligned}\label{SI:thm4-total-variance}
       \var(v_{\bm z,\bm z'}(\Zmat, \mathbf{\hat{\Psi}}, \Wmat, \Xmat)) &= \Ebb_{\Zmat,\mathbf{\hat{\Psi}}}\sbra{\var_{\Wmat, \Xmat}\pbra{v_{\bm z,\bm z'}(\Zmat,\mathbf{\hat{\Psi}}, \Wmat, \Xmat)|\Zmat, \mathbf{\hat{\Psi}}} } +\var_{\Zmat,\mathbf{\hat{\Psi}}} \pbra{ \Ebb_{\Wmat,\Xmat}\sbra{v_{\bm z,\bm z'}(\Zmat, \mathbf{\hat{\Psi}}, \Wmat, \Xmat)|\Zmat, \mathbf{\hat{\Psi}}} }.     
    \end{aligned}
 \end{align}   
To give a bound for the first term, we observe that 
\begin{align}
v_{\bm z,\bm z'} &= \frac{2^n(-1)^{\abs{\bm z} + \abs{\bm z'}}(2^{k_1+1} + 1)^{2k_2}} {N^{k_2+1}}\sum_{\bm c,\bm c'}\prod_{j = 1}^{k_2 + 1}\sum_{
\Wcal_j,\bm s_j,\bm s'_j}
(-2)^{-D(\bm s_j,\bm s'_j)}\hat p(\bm s_j,\bm c_j|\bm c_{j-1,z_{j-1}},\Wcal_j)\hat q(\bm 
 s_j,\bm c'_j|\bm c'_{j-1,z_{j-1}},\Wcal_j)\\
&=\frac{2^n(-1)^{\abs{\bm z} + \abs{\bm z'}}(2^{k_1+1} + 1)^{2k_2}} {N^{k_2+1}}\sum_{\bm s,\bm s',\Wcal}(-2)^{-D(\bm s,\bm s')}\hat{p}(\bm s|\hat{\rho}_{\bm z},\Wcal)\hat{q}(\bm s'|\hat{\sigma}_{\bm z'},\Wcal),
\label{eq:est_v_z_zprime_compress}
\end{align}  
 where $\hat{p}(\bm s|\hat{\rho}_{\bm z},\Wcal)$ is the estimation of the probability $\superbraket{\bm s|\Wcal\Ucal_{k_2+1}\hat\Psi_{k_2}\cdots \hat\Psi_1\Ucal_1|0^{n}}$, and $\hat{q}(\bm s'|\hat{\sigma}_{\bm z'},\Wcal)$ is defined analogously. Then
 \begin{align}
  \var_{\Wmat, \Xmat}\pbra{v_{\bm z,\bm z'}(\Zmat,\mathbf{\hat{\Psi}}, \Wmat, \Xmat)|\Zmat, \mathbf{\hat{\Psi}}}&= \frac{\var_{\Wcal, \Xmat}\pbra{v_{\bm z,\bm z'}(\Zmat,\mathbf{\hat{\Psi}}, \Wcal, \Xmat)|\Zmat, \mathbf{\hat{\Psi}}}}{N^{k_2 + 1}} \\
  &\leq\frac{\pbra{2^{k_1+1} + 1}^{4k_2}}{ N^{k_2+1}}\pbra{c 2^n + \frac{3^n}{m^{2(k_2 + 1)}}},   
\label{eq:v_zzprime_first_term}
 \end{align}
 for constant positive $c$.
Eq.~\eqref{eq:v_zzprime_first_term} follows from the variance bound in Eq. \eqref{eq:variance_cross_platform_general}. Note that, due to circuit cutting and cross-concatenation, we effectively obtain $m^{k_2 + 1}$ measurement outcomes.
 
For the second term, applying Eq.~\eqref{eq:est_v_z_zprime_compress}, we obtain
\begin{align}
\begin{aligned}
  &\Ebb_{\Wmat,\Xmat} \sbra{v_{\bm z,\bm z'}} = (-1)^{\abs{\bm z} + \abs{\bm z'}}(2^{k_1+1} + 1)^{2k_2}\superbra{\bm 0}\Ucal\superket{\bm 0},
\end{aligned}
\end{align}
where $\Ucal = \Ucal_1^{\dagger}\hat\Psi_1 \cdots \hat\Psi_{k_2} \Ucal_{k_2+1}^\dagger \Vcal_{k_2+1}\hat\Psi_{k_2}\cdots \hat\Psi_1\Vcal_1$.
Hence
  \begin{align}
  \begin{aligned}
  \var_{\Zmat,\mathbf{\hat{\Psi}}} \pbra{ \Ebb_{\Wmat,\Xmat}\sbra{v_{\bm z,\bm z'}(\Zmat, \mathbf{\hat{\Psi}}, \Wmat, \Xmat)|\Zmat, \mathbf{\hat{\Psi}}} } &=  (2^{k_1+1} + 1)^{2k_2}  
 \var_{\Zmat,\mathbf{\hat{\Psi}}} \pbra{(-1)^{\abs{\bm z} + \abs{\bm z'}}\superbra{\bm 0}\Ucal\superket{\bm 0}} \\
& \leq (2^{k_1+1} + 1)^{4k_2}.    
  \end{aligned}
\end{align}

Hence 
\begin{equation}
     \var(v_{\bm z,\bm z'}(\Zmat, \mathbf{\hat{\Psi}}, \Wmat, \Xmat)) \leq \frac{\pbra{2^{k_1+1} + 1}^{4k_2}}{ N^{k_2+1}}\pbra{c 2^n + \frac{3^n}{m^{2(k_2 + 1)}}} + (2^{k_1+1} + 1)^{4k_2}.
\end{equation}
Let $\hat{v}$ be the estimator defined in Eq.~\eqref{eq:cross_cutting_Lrounds_general_estimator}. By Chebyshev's inequality, we have 
\begin{align}
    \Pr\sbra{\abs{\hat{v} - \tr{\rho \sigma}}\geq \varepsilon} &\leq \frac{\var(\hat{v})}{\varepsilon^2}\\
    &=\frac{\var(v_{\bm z, \bm z'})}{L^2\varepsilon^2}\\
    &=:\delta,
\end{align}
where $\delta$ is the failure probability. We note that $\delta$ can be bounded to a small constant if the following inequalities are satisfied,
\begin{align}
L^2 &\geq c_0\frac{(2^{k_1+1} + 1)^{4k_2}}{\varepsilon^2},\\
L^2 N^{k_2 + 1}&\geq c_0c2^n\frac{(2^{k_1+1} + 1)^{4k_2}}{\varepsilon^2},\\
 L^2 m^{2(k_2 + 1)}N^{k_2+1} &\geq c_03^n \frac{(2^{k_1+1} + 1)^{4k_2}}{\varepsilon^2},
\end{align}
for some constant $c_0$.
Given these constraints, we select $L, N, m$ as follows:
\begin{align}
  L &= \Ocal\pbra{\frac{(2^{k_1+1} + 1)^{2k_2}}{\varepsilon}},\\
  N &= \Ocal\pbra{2^{n/(k_2+1)}},\\
   m &= \Ocal\pbra{(3/2)^{n/2(k_2+1)}}.
\end{align}
Hence $LNm = \Ocal\pbra{\frac{6^{n/2(k_2+1)}(2^{k_1+1} + 1)^{2k_2}}{\varepsilon}}$.
\end{proof}

\subsection{Cross-platform fidelity estimation with more general cutting}\label{SI:IV-B}

Now we extend Theorem~\ref{thm:cut_fidelity_tight} into a broader cutting framework, where the employed circuits are permitted to have more intricate configurations, as illustrated in Supplementary Figure~\ref{fig:method_alg}.

We consider $k_2$ rounds of cutting, resulting in $r$ parts after all the cuts. For simplicity, we assume that exactly $k_1$ qubits are cut in each round. However, this can easily be generalized to cases where up to $k_1$ qubits are cut. For each cut, we apply a random Clifford unitary to the cutting qubits, associated with some other operations, as illustrated in Supplementary Figure~\ref{fig:method_alg}\textbf{b}. We give a brief explanation of the algorithm as follows.

\begin{figure*}
    \includegraphics[width = 0.6\textwidth]{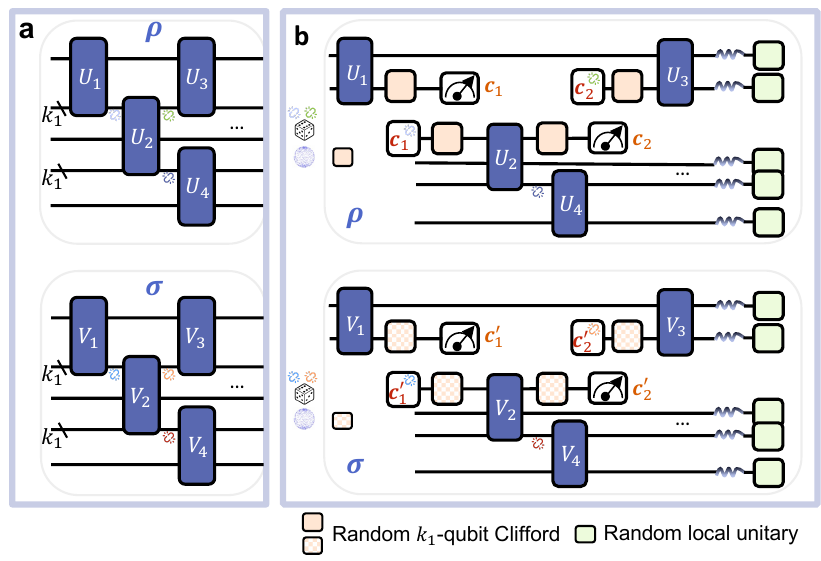
}
    \caption{Illustration of the quantum circuit used to estimate $\tr{\rho\sigma}$ via the state similarity algorithm with classical communication. \textbf{a.} State preparation circuits for the quantum states $\rho$ and $\sigma$. \textbf{b.} Quantum circuit implementing the state similarity algorithm with classical communication. Here, the orange-colored boxes denote random Clifford gates acting on $k_1$ qubits, and green-colored boxes represent the tensor product of random local unitary $\Wcal$. For clarity, only two circuit cuts are explicitly shown as an example. Wavy lines preceding the random unitaries $\mathcal{W}$ indicate the omission of intermediate gates for simplicity. Final measurements are performed on the computational basis following the application of $\mathcal{W}$.}
\label{fig:method_alg}
\end{figure*}

Generate the local random unitary set $\Wmat = \cbra{\Wcal^{(1)},\ldots, \Wcal^{(N)}}$, where each $\Wcal\in \Wmat$ is a tensor product of local random unitaries. Here, with a bit of abuse of symbols, we utilize $\Wcal_j$ to denote the random unitary associated with \textsf{Q}-$j$, i.e.,  $\Wcal = \otimes_{j = 1}^{r} \Wcal_j$ and $\Wmat_j$ to denote the associated set, as illustrated by the green-colored boxes in each independent part of Supplementary Figure~\ref{fig:method_alg}\textbf{b}.

The algorithm settings and circuit structure remain the same as in the previous section, except for the operations on cutting qubits.

Recall that a cut consists of a \textit{incoming cut} and a \textit{departing cut}. In \textsf{Q}-$j$, suppose there are $u_j$ \textit{departing cuts}, and denote the set of these cuts as $A_j$. Consequently, the re-initialized state will involve $u_j \times k_1$ qubits, which we denote by $\bm{c}_{\text{pre}(j), \bm{z}}$, where $\text{pre}(j)$ represents all parts containing the associated $u_j$ number of \textit{incoming cuts}.  
The re-initialized inputs corresponding to the cut $v \in A_j$ are given by the measurement outcome of the associated \textit{incoming cut} if $z_v=0$, and approximate $\superket{\Ibb}$ if $z_v=1$.  
For each cut $v$, a Clifford gate $\mathcal{C}^\dagger$ is applied to the \textit{incoming cut} before computational basis measurements, while the inverse Clifford gate $\mathcal{C}$ is applied after re-initialization for the \textit{departing cuts}, with probabilities determined by the associated Bernoulli variable $z_v$. All other gates remain unchanged. Finally, a random unitary $\mathcal{W}$ is applied prior to the $Z$-basis measurements, as illustrated in Supplementary Figure~\ref{fig:method_alg}\textbf{b}.
  
We define $p\pbra{{\bm s}_j,\bm c_j \mid \bm{c}_{\text{pre}(j),\bm z}, \mathcal{W}_j}$ as the probability of obtaining outcome $({\bm s}_j, \bm c_j)$ in \textsf{Q}-$j$, where $\bm c_j$ denotes the measurement outcomes of the \textit{incoming cuts}, and $\bm s_j$ denotes the measurement outcomes of the remaining outcomes, conditioned on the measurement outcomes and input configurations of the \textit{incoming cuts}, and given the application of the random unitary channel $\mathcal{W}_j$. Similar to $p\pbra{{\bm s}_j,\bm c_j \mid \bm{c}_{\text{pre}(j),\bm z}, \mathcal{W}_j}$, we define $q(\bm s'_j,\bm c'_j \mid \bm{c}'_{\text{pre}(j),\bm z}, \mathcal{W}_j)$ as the corresponding probability associated with the quantum state $\sigma$ prepared on Platform 2. Let $\hat{p}(\cdot)$ and $\hat{q}(\cdot)$ be the estimated probabilities of $p(\cdot)$ and $q(\cdot)$ respectively, obtained from $m$ rounds of measurement repetitions. Following these notations, the estimator of $\tr{\rho\sigma}$ is as follows
  \begin{align}
  \begin{aligned}
 &v_{\bm z,\bm z'} :=\frac{2^n(-1)^{\abs{\bm z}+\abs{\bm z'}}\pbra{2^{k_1+1} + 1}^{2k_2} }{N^{r}}\sum_{\bm c, \bm c'}\prod_{j = 1}^{r}\sum_{
\Wcal_j,\bm s_j,\bm s'_j}(-2)^{-D(\bm s_j,\bm s'_j)}\hat p(\bm s_j,\bm c_j \mid \bm{c}_{\text{pre}(j),\bm z}, \mathcal{W}_j)\hat q(\bm s'_j,\bm c'_j \mid \bm{c}'_{\text{pre}(j),\bm z'}, \mathcal{W}_j),
  \end{aligned}
\label{eq:cross_cutting_general_estimator1}
\end{align}  
where $\bm c:=\pbra{\bm c_1,\ldots \bm c_{r}} $ and $ \bm c'=\pbra{\bm c'_1,\ldots, \bm c'_{r}}$.
Let $\Zmat_1 = (\bm z_1, \bm z_2, \ldots, \bm z_{L})$ and $\Zmat_2 = (\bm z_1', \bm z_2', \ldots, \bm z_{L}')$ be sets of randomly generated Bernoulli variables. For each pair $\bm z \in \Zmat_1$ and $\bm z' \in \Zmat_2$, we generate $v_{\bm z, \bm z'}$ using Eq.~\eqref{eq:cross_cutting_general_estimator1} with randomly generated unitaries $\Ccal$ and $\Ccal'$ corresponding to $\bm z$ and $\bm z'$, respectively. The final estimator is then given by 
\begin{align}
  \hat v = \sum_{\bm z,\bm z'} \frac{v_{\bm z,\bm z'}}{L^2}.
\label{eq:cross_cutting_general_estimator2}
\end{align}

Equation~\eqref{eq:cross_cutting_general_estimator2} provides an unbiased estimator of $\tr{\rho\sigma}$, and the required number of samples can be reduced sub-exponentially, provided that $k_1$ and $k_2$ are constants with $k_1 < r$ and $r \geq 3$. This result, along with Theorem 1 in the main text, is established by leveraging the 3-design property of the Clifford group and tools from the theory of random unitaries.
Hence, the number of measurements can be upper bounded by $\Ocal\pbra{\pbra{2^{k_1+1} + 1}^{2k_2} 6^{n/2r}/\varepsilon}$, as formally stated in the following theorem, which corresponds to Theorem~1 in the main text.

\begin{theorem}[Formal version of Theorem 1 in the main text]
Let $\hat{v}$ be the estimator of $\tr{\rho \sigma}$ defined in Eq.~\eqref{eq:cross_cutting_general_estimator2}, where $\rho$ and $\sigma$ are quantum states. Suppose each cut involves at most $k_1$ qubits, and there are $k_2$ total cuts dividing the circuit into $r$ parts.
Then, $\hat{v}$ is an unbiased estimator of $\tr{\rho \sigma}$, and satisfies $\abs{\hat{v} - \tr{\rho \sigma}} \leq \varepsilon$ with high probability, provided the number of measurements is upper bounded by
\[
\mathcal{O} \left( \left(2^{k_1 + 1} + 1\right)^{2k_2} \cdot \frac{6^{n / 2r}}{\varepsilon} \right).
\]
\label{thm:cut_fidelity_general}   
\end{theorem}
The proof is similar to Theorem \ref{thm:cut_fidelity_tight} by generalizing the sequential cutting with generating $k_2+1$ parts into more generating cutting with generating $r$ independent parts.

\subsection{Cross-platform fidelity for noiseless and noisy graph states with measurement error mitigation}

Here we give the analysis for the fidelities between the generated noisy graph state with circuit cutting and the noiseless graph state $\tr{\hat{\rho}\rho_{G}}$, where $G$ denotes the graph representation for the graph state and $\hat{\rho}$ is the prepared noisy quantum state.

Let us start by briefly introducing the graph state representation for $n$-qubit quantum states. An $n$-qubit graph state $\rho_G$ can be written as
\begin{align}
\rho_G = \prod_{j}\frac{\pbra{\mathbb{I} + S_j}}{2} = \frac{1}{2^n} \sum_P \text{sign}(P)P,
\end{align}
where $S_j = X_j \prod_{l \in N(j)} Z_l$ is the stabilizer generator for qubit $j$, and $N(j)$ denotes the set of neighbors of vertex $j$ in the graph $G$. Each $P$ in the sum is a stabilizer element with $\text{sign}(P) \in \cbra{\pm 1}$. In particular, for the GHZ state, the stabilizers are given by $S_1 = X_1 X_2 \ldots X_n$ and $S_j = Z_{j-1} Z_j$ for $j \in \cbra{2, \ldots, n}$.

By the Chernoff bound, $\mathcal{O}(1)$ samples are sufficient to estimate $\tr{\hat{\rho} \rho_G}$ with constant additive error and high probability, when $P$ is sampled uniformly at random from the stabilizer group~\cite{Cao2023GenerationOG}. In the following, we present the estimation of $\tr{\hat{\rho} \rho_G}$ for a noisy quantum state prepared via the circuit cutting method, and introduce an error mitigation strategy to suppress readout errors.

\subsubsection{Algorithm details}

 We assume that only a single qubit is cut, partitioning the quantum state into two parts: \textsf{Q}-\textsf{A} and \textsf{Q}-\textsf{B}. This setup can be readily extended to the general case with multiple cuts. 
Let $\text{sign}(\mathcal{P}_t)\mathcal{P}_t$ denote the randomly generated stabilizer for each $t \in [T]$, where $T$ is the total number of stabilizers sampled.

Recall the definition of the estimator for the cross-platform fidelity with a single cut in SI \ref{subsec:in_parallel_method}, where we define the set of tuples
\begin{align}
    \mathcal{Z} = \cbra{\pbra{0, \mathbb{I}, \frac{1}{5}}, \pbra{0, R_y^{\pi/2}, \frac{1}{5}}, \pbra{0, R_x^{\pi/2}, \frac{1}{5}}, \pbra{1, \mathbb{I}, \frac{2}{5}}}
\end{align}
to simplify circuit configurations associated with cutting. Here we adopt the same strategy and introduce $\mathcal{Z}$ for calculating the fidelity $\tr{\rho \rho_G}$, where $\rho$ is the prepared quantum state on Platform~1 (or Platform~2) with wire cutting, and $\rho_G$ is the ideal graph state.  An unbiased estimator for this fidelity can be expressed as
\begin{align}
\hat{v}_{G} &= \frac{5}{T} \sum_{j=1}^4 (-1)^{\Zcal_{j,1}} \Zcal_{j,3}  
\sum_{t = 1}^T\text{sign}(\Pcal_t) \sum_{c\in\cbra{0,1}} \sum_{\bm a} \mu(\bm a)\hat{p}\pbra{\bm a, c | \Zcal_j, \Wcal_t^{\Pcal_{\bar A}}} \sum_{\bm b} \mu\pbra{\bm b} \hat{p}\pbra{\bm b|\Zcal_j, c, \Wcal_t^{\Pcal_B}},
\label{eq:estimator_pure_state}
\end{align}
where $\Wcal^{\Pcal}:= \Wcal^{\Pcal_{\bar A}} \otimes \Wcal^{\Pcal_B}$ denotes a randomly chosen stabilizer unitary that generates the associated graph state, $\bar A$ denotes the qubits in $A$ except the cutting qubits that satisfies the linear map $P_t = W_t^\dagger Z^{\otimes n} W_t$, and $\Pcal_t$ and $\Wcal_t$ are channel representations of unitaries $P_t$ and $W_t$ respectively, $\mu(\bm s) = (-1)^{\sum_{j=1}^{|\bm s|}s_i}$, and $\abs{\bm s}$ is the size of $\bm s$.
Bitstring $(\bm{a}, c)$ denotes the measured bitstring in \textsf{Q}-\textsf{A}, and $\bm{b}$ denotes the measured bitstring in \textsf{Q}-\textsf{B}, consistent with the definitions provided in the Methods section. The notations $c$ and $\Wcal_t$ are also used as defined in the Methods section and are not repeated here.

When we refer to an unbiased estimation of the cross-platform fidelity or the fidelity with a pure state, we assume that the random unitary $\mathcal{W}$ and the measurement process are noiseless. In the following, we introduce the error mitigation algorithm designed to mitigate the measurement noise.

\smallskip
\noindent\textbf{Error mitigation for cutting associated noise.}
We focus on mitigating measurement noise on the cutting qubits, under the assumption that state preparation and the Clifford gates applied to the cutting qubits are noise-free. This assumption is reasonable in the single-qubit cut setting, where measurement noise dominates, and errors from zero-state preparation and single-qubit gates are negligible by comparison~\cite{Arute2019Quantum}.

Since the channel $\Phi_1$ defined in Eq.~\eqref{eq:phi_1} resets the input of the cutting qubits, we only need to consider the noise introduced by the protocol of $\Phi_0$ defined in Eq.~\eqref{eq:phi_0}. Assuming the random Clifford gates and state preparation are noise-free, we denote the noisy channel $\Phi_0$ as $\widetilde{\Phi}_0$, which can be expressed as:
\[
\widetilde{\Phi}_0 = \Ebb_{U}\sbra{\Ucal \circ \widetilde{M}_{\bm s} \circ \Ucal^{\dagger}},
\]
where $\widetilde{M}_{\bm s} = \Lambda \sum_{\bm s} \superket{\bm s}\superbra{{\bm s}}$ represents the noisy measurement and noiseless state preparation processes, $\Lambda$ is the measurement noise channel.
By a corollary of Schur's lemma~\cite{fulton2013representation},  we have 
\begin{align*}
  \widetilde{\Phi}_0 = \superket{\Ibb_{d}}\superbra{\Ibb_{d}} + f\Pi_1,  
\end{align*}
where $f = \tr{\widetilde{M}_{\bm s}\Pi_1}/\tr{\Pi_1}$, and $d = 2$.
By the definition of $\Pi_1$, we have $\Pi_1 = \Ical - \superket{\Ibb}\superbra{\Ibb}$ with  $\tr{\Pi_1} = d^2 - 1$. 
Therefore, $f = \frac{1}{d+1}$ if the quantum device is noise-free.

We approximate $f$ using the calibration protocol proposed by Chen et al.~\cite{Chen2021Robust}. To ensure completeness, we restate the protocol below and give the illustration in Supplementary Figure~\ref{suppfig:errmiti_cutting}. 
We repeat $T$ rounds of the following steps to approximate $f$ by $\hat{f} = \frac{1}{T}\sum_{t=1}^T \hat{f}_t$, where we generate an estimation $\hat{f}_t$ in each round.
\begin{itemize}
    \item[(1)] Randomly select a Clifford element $\Ccal$ from the $\Clgroup_{1}$, apply $\Ccal^\dagger$ to the readily prepared state $\ket{0}$, and then measure it in the computational bases, yielding outcomes $\bm s$.
    \item[(2)] Generate estimation $\hat{f}_t = \frac{d\tr{\ket{\bm 0}\bra{\bm 0} \Ccal \ket{\bm s}\bra{\bm s}\Ccal^\dagger }  - 1}{d - 1}$.
\end{itemize}

\begin{figure}
    \centering
    \includegraphics[width=0.8\linewidth]{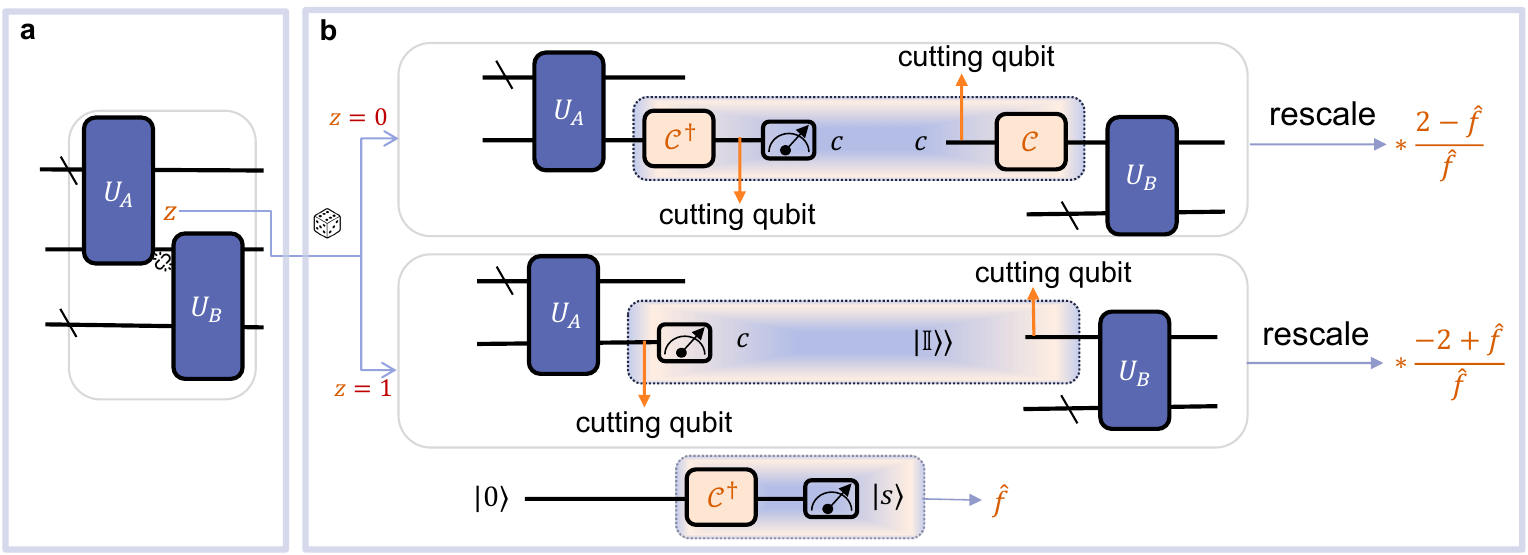}
    \caption{Schematic illustration of the circuit-cutting algorithm with error mitigation of the cut-qubit measurement noise. \textbf{a.} The circuit before the cut. \textbf{b.} The sub-circuits after the cut and the calibration circuit. The bottom calibration process generates $\hat{f}$, which serves for the estimation process.}
    \label{suppfig:errmiti_cutting}
\end{figure}

Here we prove that $\Ebb[\hat{f}] = f$ using the twirling property of the Clifford group and calculate it within the Pauli-transfer matrix representation, as shown in the following,
\begin{align*}
\Ebb\sbra{\hat{f}} &= \frac{d\Ebb_{\Ccal}\sum_{\bm s}\sbra{\superbra{\bm 0}\Ccal\superket{\bm s}\superbra{\bm s}\Lambda\Ccal^\dagger\superket{\bm 0}} - 1}{d - 1} \\
&= \frac{d \superbra{\bm 0} \pbra{\superket{\Ibb}\superbra{\Ibb} + f\Pi_1}\superket{\bm 0} - 1}{d - 1} \\
&= \frac{d\pbra{f + (1-f)\frac{1}{d}} - 1}{d-1}\\
&=f,
\end{align*}
where $\Ccal = C(\cdot )C^{\dagger}$ is the channel representation of the Clifford gate $C$ and $d = 2$, and $\Lambda$ is the noise channel.

With the approximated $f$, denoted as $\hat{f}$, we let $z\in\cbra{0,1}$ be the Bernoulli random variable such that $\Pr[z = 0] = \frac{1}{2-\hat{f}}$ and replace $\Phi_0$ by $\widetilde{\Phi}_0$, we have $\Ical = \frac{2-\hat{f}}{\hat{f}}\Ebb_z\sbra{(-1)^z\widetilde{\Phi}_z}$, where $\widetilde{\Phi}_i$ denotes $\Phi_i$ with measurement noise. 
We can easily check that for the noise-free model,
\begin{align*}
    \Pr[z = 0] = \frac{1}{2 - f} = \frac{d+1}{2d + 1}, \qquad \Ical = \frac{2-{f}}{{f}}\Ebb_z\sbra{(-1)^z{\Phi}_z} = \frac{d + 1}{2d+1} \Ebb_z\sbra{(-1)^z{\Phi}_z},
\end{align*}
where $d = 2$, which matches the noise-free result in the previous section.
Consequently, the number of measurements is associated with $\hat{f}$. 
 
As a contrast to the distribution of circuit configuration defined in SI \ref{subsec:in_parallel_method}, we let the distribution of the circuit configuration as 
\begin{align}
\hat\Zcal^{(i)} = \left\{\left( 0, \Ibb, \frac{1}{3\pbra{2 - \hat{f}_i}} \right), \left( 0, R_y^{\pi/2}, \frac{1}{3\pbra{2 - \hat{f}_i}} \right), \left( 0, R_x^{\pi/2}, \frac{1}{3\pbra{2 - \hat{f}_i}} \right), \left( 1, \Ibb, \frac{{1 - \hat{f}}}{{2 - \hat{f}_i}} \right)\right\},
\end{align}
where $\hat{f}_i$ is an estimator of $f_i$ with the above calibration process, $i$ denotes Platform $i$, and $\hat\Zcal^{(i)}$ is the configuration set associated with the Platform $i\in\cbra{1,2}$.

Since here we consider a single cut of a quantum state with $k_1=k_2=1$, then for a quantum state $\rho$, after cutting, similar to Eq.~\eqref{eq:estimator_pure_state}, we denote the estimator of $\tr{\hat{\rho}\rho_G}$ obtained through the error-mitigated cutting process as
\begin{align}
\hat{v}_{G,z}&:=\frac{ 2^n }{T}
\frac{\pbra{2-\hat{f}}}{\hat{f}}\sum_{j} (-1)^{\Zcal_{j,1}}\Zcal_{j,3} \sum_{t = 1}^T\text{sign}(\Pcal_t) \sum_{c}\sum_{\bm a}\mu(\bm a)\hat{p}\pbra{\bm a, c | \Zcal_j, \Wcal_t^{\Pcal_{\bar A}}} \sum_{\bm b}  \mu\pbra{\bm b} \hat{p}\pbra{\bm b|\Zcal_j, c, \Wcal_t^{\Pcal_B}},
\label{eq:em_cutting}
\end{align}
where $\Wcal^{\Pcal}:= \Wcal^{\Pcal_{\bar A}} \otimes \Wcal^{\Pcal_B}$ denotes a randomly chosen stabilizer unitary that generates the associated graph state, satisfying the linear map $P = W^\dagger Z W$, $\hat{f}$ is the estimation of $f$, which corresponds to the cutting of Platform 1 or 2.

\smallskip
\noindent\textbf{Error mitigation for readout error of other qubits.}
Assuming that the measurement noise exhibits no cross-talk and can be modeled as a tensor product of $n$ independent two-level noise channels, the overall noise channel $\Lambda$ takes the form
\begin{align}
\Lambda = \bigotimes_{j=1}^n \Lambda_j, \quad \Lambda_j = \begin{pmatrix} 1-\xi_j & \eta_j\\
\xi_j & 1 - \eta_j \end{pmatrix},
\end{align} where each $\Lambda_j$ represents the local noise acting on the $j$-th qubit, and the parameters $\xi_j$ and $\eta_j$ correspond to the probabilities of bit-flip errors from $0$ to $1$ and from $1$ to $0$, respectively. Under this assumption, we give an unbiased estimation
\begin{align}
    \hat{v}_{G}= \frac{5}{T} \sum_{j=1}^4 (-1)^{\Zcal_{j,1}} \Zcal_{j,3} \sum_{t = 1}^T \text{sign}(\Pcal_t)\sum_{c}
\sum_{\bm x,\bm a} 
 \hat{p}\pbra{\bm a, c | \Zcal_j, \Wcal_t^{\Pcal_{\bar A}}} \mu(\bm x)\braket{\bm x|\Lambda^{-1}|\bm a}
 \sum_{\bm y,\bm b}
 \hat{p}\pbra{\bm b|\Zcal_j, c, \Wcal_t^{\Pcal_B}} \mu\pbra{\bm y} \braket{\bm y|\Lambda^{-1}|\bm b}.
 \label{eq:em_readout}
\end{align}

This quantity $\hat{v}$ generates an unbiased estimation of $\tr{\rho\ket{\psi}\bra{\psi}}$ by Ref.~\cite{bravyi2021mitigating}.
Combined with Eq.~\eqref{eq:em_cutting} and Eq.~\eqref{eq:em_readout}, we generate the final error-mitigated estimator.

\section{Details of superconducting quantum processors}
\label{supp:intro_supconduct_device}

Our experiment is conducted on a multi-qubit superconducting quantum processor, similar in architecture to that described in Ref.~\cite{liang2024dephasing}. The processor consists of 66 frequency-tunable transmon qubits and 110 tunable couplers, each connecting neighboring qubits. In the GHZ cross-platform fidelity estimation experiment, 19 frequency-tunable asymmetric qubits are used, while 18 qubits are used in the quantum phase learning experiment. As these two sets of qubits share 16 qubits, a total of 22 distinct qubits are utilized across both experiments.
The sweet points of the qubits are around 3.5 GHz and 7.0 GHz, with nonlinearity around -220 MHz. 
Each qubit is equipped with a dedicated Z-bias line for frequency tuning and implementation of Z gates, as well as an independent XY-control line for the application of XY gates.
The tunable couplers are individually controlled via separate Z-bias lines and enable dynamic modulation of the effective coupling strength between adjacent qubits, ranging from approximately $+5$ MHz to $-20$ MHz. A subset of 22 qubits is selected to experimentally demonstrate state similarity with communication.

Key-performance characteristics of the utilized qubits are summarized in Supplementary Figure~\ref{suppfig:parameters}. The idle frequencies of the qubits range from approximately 4.0 GHz to 4.5 GHz. The average relaxation time $T_1$ and spin-echo dephasing time $T_{2}^\text{SE}$ are measured to be 76.7~$\mu$s and 11.7~$\mu$s, respectively.
Single-qubit gate fidelities are assessed using randomized benchmarking~\cite{knill2008randomized}, yielding an average gate error of 0.08\%. Two-qubit CZ gates are characterized via cross-entropy benchmarking~\cite{boixo2018characterizing, google2023suppressing, ren2022experimental}, with a measured average cycle Pauli error of 1.14\%. The average readout fidelities for the $\ket{0}$ and $\ket{1}$ states are 98.0\% and 95.1\%, respectively.
\begin{figure}[htpb]
    \centering
    \includegraphics[width=1.0\linewidth]{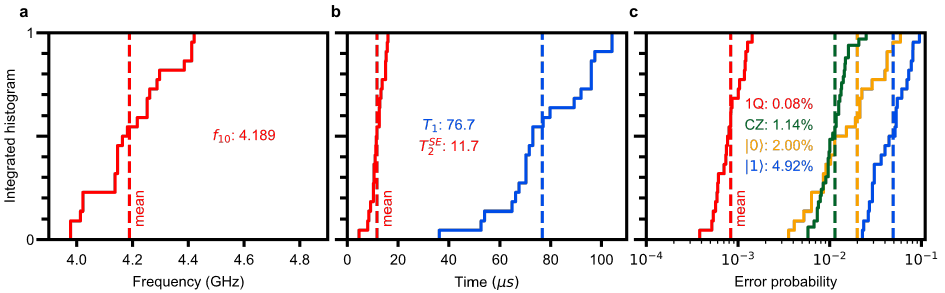}
    \caption{Integrated histograms of typical parameters for the employed 22 qubits. \textbf{a.} Distribution of qubit idle frequency $f_\text{10}$. \textbf{b.} Distribution of qubit relaxation time $T_1$(read line) and spin echo dephasing time $T_2^\text{SE}$(blue line) measured at idle frequency. \textbf{c.} Pauli error for single-(red line) and the cycle Pauli error for two-qubit CZ gates(green line). The yellow and blue lines indicate the readout errors for the $\ket{0}$ and $\ket{1}$ states, respectively.}
    \label{suppfig:parameters}
\end{figure}

\section{More Experimental Results}
\label{supp:experimental_res}
In this section, we provide a more detailed explanation of the experimental setup described in the main text, along with additional supporting results. The experiments involve a single circuit cut applied to one qubit. We also present a numerical analysis of multi-cut scenarios, showing that as the number of circuit partitions increases, the estimation error decreases for a fixed number of measurements.

\subsection{Experimental implementation of cross-platform fidelity estimation of GHZ states with circuit cutting}

\subsubsection{Parameter settings for the GHZ-state experiments in the main text}
Supplementary Table~\ref{tab:GHZ_param_set} summarizes the number of random unitaries applied for each qubit configuration, along with the measurement shots used in the experiments. Recall that for GHZ-state experiments, a single cut is applied to the $(n+1)/2$-th qubit, dividing the quantum state into two parts. Accordingly, in the `Qubits’ column, the notation $(j+j)$ indicates that each device contains $j$ qubits. To mitigate statistical fluctuations and estimate error bars, each experiment is repeated over 12 slightly dependent sub-experiments. The reported value in the main text is the mean of these 12 estimations, and the error bar is given by the standard deviation divided by $\sqrt{12}$.

\begin{table}[h]
\caption{Parameter settings for the cross-platform fidelity estimation with classical communication for the GHZ state shown in Fig. 2 of the main text. 
}
\begin{tabular}{|c|c|c|c|c|c|c|}
\hline
\textbf{State} & \textbf{Qubits} & \textbf{\#$\Wcal$} &\textbf{\#Circ. a $\Wcal$} & \textbf{\#Shots a Circ.}  &  \textbf{\#$\Wcal$ a Sub-exp.} & \textbf{\#Shots a Circ. in a Sub-exp.} \\ \hline
\multirow{4}{*}{GHZ} & 5 ($3+3$) 
& \textsf{Q}-\textsf{A}:9/\textsf{Q}-\textsf{B}:27 & \multirow{4}{*}{12} & 6000  
& \textsf{Q}-\textsf{A}:9/\textsf{Q}-\textsf{B}:27 & 500\\ 
\cline{2-3}  \cline{5-7} 
 & 7 ($4+4$) & 80 & & 6000  & 40 & 1000 \\ \cline{2-3}  \cline{5-7} 
 & 9 ($5+5$) & 100  & & 6000  & 50& 1000 \\ \cline{2-3}  \cline{5-7} 
 & 11 ($6+6$) & 400 &  & 2000  & 66& 1000 \\ \hline
\end{tabular}
\label{tab:GHZ_param_set}
\end{table}

To enhance the performance of the superconducting device in our experiment, we conduct all 12 sub-experiments simultaneously and subsequently partition the data into 12 sub-experiments, thereby enabling data reuse across these sub-experiments. Accordingly, we report (1) the total number of $\mathcal{W}$ unitaries (in the ``\# $\mathcal{W}$'' column), (2) the number of circuits induced by the cutting process (in the ``\# Circ. a $\mathcal{W}$'' column), and (3) the number of shots per circuit (in the ``\# Shots a Circ.'' column).

The ``\# $\mathcal{W}$'' column lists the number of random unitaries $\mathcal{W}$ applied to each part for $n \in \{7,9,11\}$, while all unitaries are enumerated when $n=5$. Recall that we have two parts, and we distinguish them as \textsf{Q}-\textsf{A} and \textsf{Q}-\textsf{B} in the $n=5$ case, as they involve different numbers of enumerated unitaries.

As detailed in the Methods section, we parallelize the quantum execution and classical postprocessing by enumerating all circuit configurations associated with a given cut, replacing the sampling of Bernoulli variables and cutting channels with full enumeration. Consequently, for each entry in the ``\# Circ. a $\mathcal{W}$'' column, 12 times more circuits are executed to account for all combinations of cutting configurations, as shown in the last column of Supplementary Table~\ref{tab:GHZ_param_set}.

After dividing the data from an experiment into 12 sub-experiment datasets, we list the number of random unitaries applied in each sub-experiment (in the ``\# $\mathcal{W}$ a Sub-exp.'' column) and the total number of shots used per sub-experiment (in the ``\# Shots a Circ. in a Sub-exp.'' column).

Generating $N$ random unitaries allows for the construction of $N^2$ distinct circuit configurations, which significantly reduces the overall measurement cost when the number of qubits exceeds 7. In the 5-qubit case, we enumerate all possible random unitaries: the upper subcircuit involves 2 qubits and yields 9 unitaries, while the lower subcircuit involves 3 qubits and yields 27 unitaries. As a result, the total number of regenerated unitaries is $9 \times 27 = 243$.

\begin{figure}[htpb]
    \centering
    \includegraphics[width=0.8\linewidth]{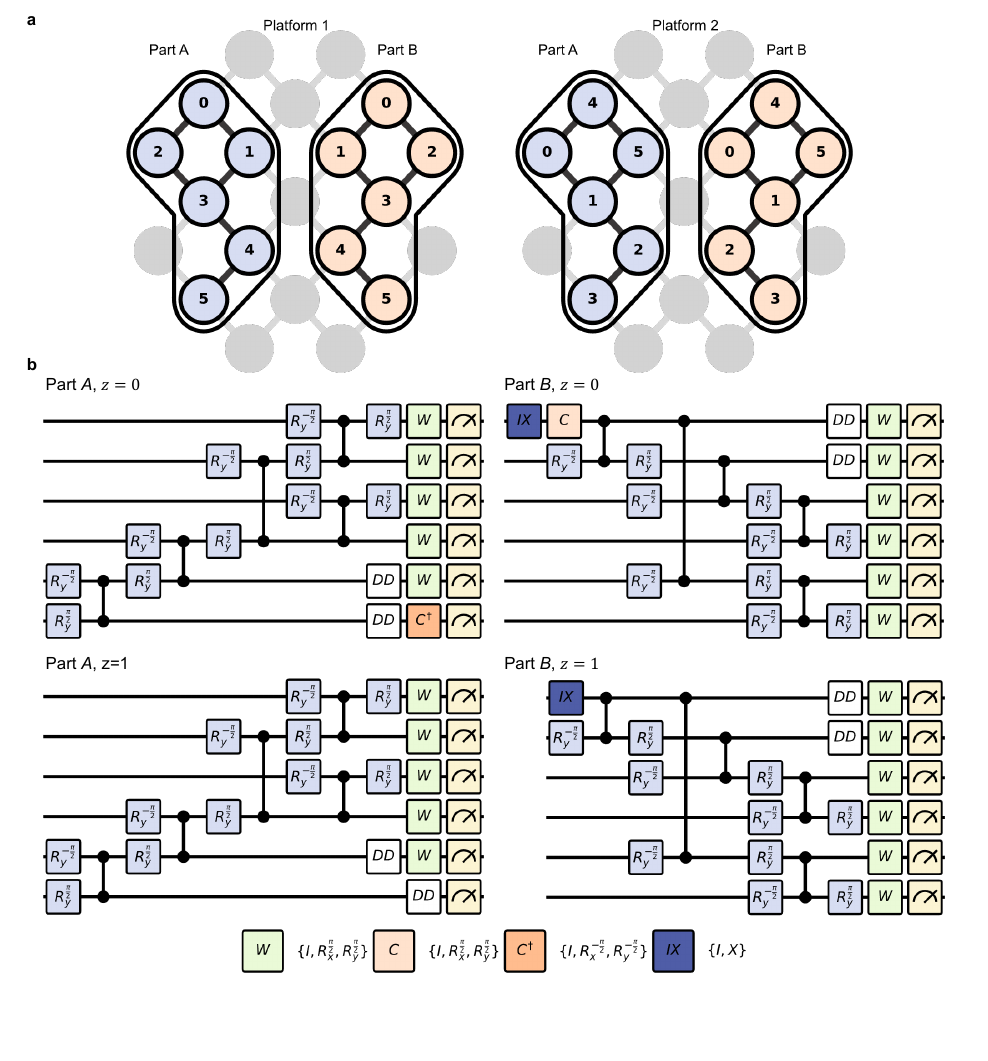}
    \caption{Experimental circuit for cross-platform fidelity estimation of 11-qubit GHZ states. DD in the circuit refers to the use of dynamical decoupling techniques to mitigate decoherence. 
\textbf{a.} Physical qubits and their connectivity used in the experiment, with two platforms employed—each comprising two parts (modules). 
\textbf{b.} Specific circuits after cutting, as implemented in the experiment. The quantum state is divided into two parts: \textsf{Q}-\textsf{A} and \textsf{Q}-\textsf{B}, with both platforms executing the same circuit. In \textsf{Q}-\textsf{A}, 4 circuit configurations are performed, while in \textsf{Q}-\textsf{B}, 8 configurations are executed.  
}
    \label{suppfig:SI_6+6circuits}
\end{figure}
\subsubsection{Experimental circuit settings in the main text}

In all of the experiments, we utilized the same quantum processor and selected two sets of qubits, labeled Platform 1 and 2, each consisting of $(n+1)/2$ qubits, to implement the quantum circuit. 
For \textsf{Q}-\textsf{A}, single- and two-qubit gates were applied to prepare an $n$-qubit GHZ state. When $z=0$, a quantum operation was randomly chosen from $\{ I, R_x^{\frac{\pi}{2}}, R_y^{\frac{\pi}{2}} \}$ and applied to the qubits before measurement, while for the cutting qubit, the operation was randomly selected from $\{ I, R_x^{-\frac{\pi}{2}}, R_y^{-\frac{\pi}{2}} \}$. 
When $z = 1$, no additional quantum gate is performed on the cutting qubit before measurement.

For \textsf{Q}-\textsf{B}, since we have parallelized the circuit with additional postprocessing, the qubit must first undergo the preparation of either the $|0\rangle$ or $|1\rangle$ state. For $z=0$, a quantum operation randomly chosen from $\{ I, R_x^{\frac{\pi}{2}}, R_y^{\frac{\pi}{2}} \}$ is applied to the cutting qubit after state preparation. However, when $z = 1$, no additional operation is needed for the cutting qubit.

The circuit for preparing the 11-qubit GHZ state ($6+6$ qubits) via circuit cutting is shown in Supplementary Figure~\ref{suppfig:SI_6+6circuits}.
Depending on the qubits' idle time, we applied the dynamical decoupling (DD) technique~\cite{ezzell2023dynamical} to mitigate decoherence.

\subsubsection{More experimental results for GHZ-states}
To assess the quality of the generated GHZ state, we estimate its fidelity with the ideal (noiseless) GHZ state using measurement data obtained from the cross-platform fidelity experiment, as indicated in FIG. 2c in the main text.
This is feasible because we use local Clifford gates, which diagonalize Pauli operators. That is, if a Pauli operator $P$ admits the decomposition $P = C^\dagger \Lambda C$ for some Clifford $C$, then we can estimate $\tr{\rho P}$ by applying $C$ to the state $\rho$, performing a computational basis measurement, and computing the estimator $\mu(s) = (-1)^{\sum_j \bm{s}_j}$ from the measurement outcome $\bm{s} \in \cbra{0,1}^n$, which provides an unbiased estimate of $\tr{\rho P}$.

To estimate the fidelity with the pure state, we face the challenge that the quantum state is destroyed by the measurements. Moreover, due to the system noise, it is difficult to regenerate an identical state. To address this, we reuse the measurement results from the cross-platform fidelity experiment to estimate the fidelity of the generated quantum states with the ideal pure state on both Platform 1 and Platform 2.

One approach is to directly select the global random unitary set $\cbra{\Wcal}$ that forms the Pauli operators belonging to the stabilizer group.   Since the set of $2^n$ stabilizer generators forms a small subset of the full $4^n$ Pauli basis, this chosen set is relatively small. Importantly, both local and global Pauli stabilizers contain $2^n$ elements, meaning their total numbers are identical.
To satisfy the randomness requirement and expand the measurement statistics, we choose an equal number of global stabilizers and local stabilizers for comparison. Local stabilizers can be easily selected from the $\cbra{\mathcal{W}}$ set by using Pauli measurement methods~\cite{wu2023overlapped}.
In practice, we first identify the number of global stabilizers present among the random Pauli operators generated by the random unitaries. We then select an equal number of local stabilizers from the set of random Pauli operators for comparison.  

\begin{figure}[htpb]
    \centering
    \includegraphics[width=1.0\linewidth]{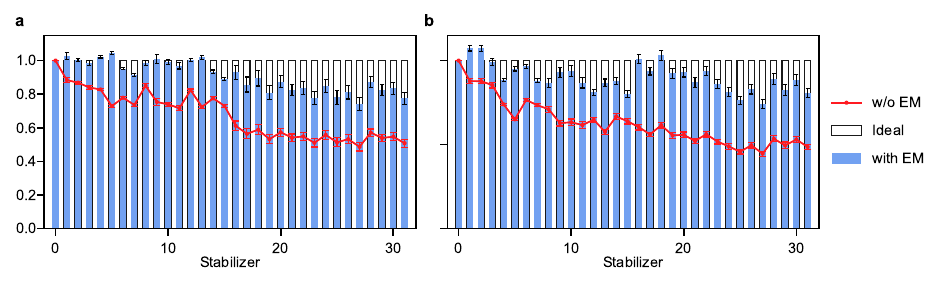}
    \caption{Expectation values of all stabilizer operators for 5-qubit GHZ states cut into $3+3$ parts, with and without error mitigation.
Results are shown for GHZ states prepared on \textbf{a.} Platform~1 and \textbf{b.} Platform~2.
}
\label{suppfig:SI_stabilizer_expectation}
\end{figure}

In Supplementary Figure~\ref{suppfig:SI_stabilizer_expectation}, we present a comparison of all stabilizer expectation values of the 5-qubit GHZ state in two platforms without and with error mitigation proposed in SI Section \ref{supp:algorithm_details}, based on experimental data. The GHZ state is generated with a single cut for the third qubit.

\begin{figure}[t]
    \centering
    \includegraphics[width=0.4\linewidth]{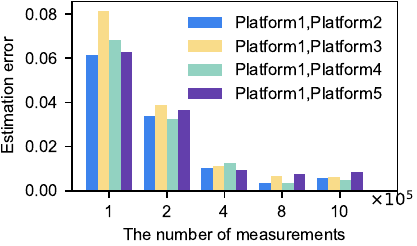}
    \caption{Estimation error as a function of the number of measurements. The error is computed from 12 independent rounds of the 9-qubit GHZ cross-platform fidelity estimation experiment, using the same experimental measurement data as in Figure 2 of the main text.}
    \label{suppfig:ErrMeas}
\end{figure}

We note that the estimation of tomography results with Pauli operators in the main text also follows the same approach by enumerating the first 64 Pauli operators. Each Pauli operator is encoded using a quaternary number, where the digits 0, 1, 2, and 3 represent $\mathbb{I}$, $\sigma_x$, $\sigma_y$, and $\sigma_z$, respectively. The first 64 Pauli operators correspond to the first 64 quaternary numbers under this encoding scheme.

Supplementary Figure~\ref{suppfig:ErrMeas} illustrates how the estimation error scales with the number of measurements, using the same data as in Figure 2 of the main text for the 9-qubit ($5+5$) cross-platform fidelity experiment. The corresponding settings for the number of random unitaries, measurement shots per round, and the number of repetition rounds used for error estimation are provided in Supplementary Table~\ref{tab:fiveplusfive_settings}. The error is defined as the standard deviation over the repeated rounds.

Supplementary Figure~\ref{suppfig:5plus5_fidpure} shows the fidelity between the noisy 9-qubit GHZ state and the ideal GHZ state, obtained by cutting the 5th qubit and implementing the circuit on two independent 5-qubit modular quantum processors. This setup can also be realized on a single modular quantum device. The figure compares results with and without error mitigation, using the method described in SI Section~\ref{supp:algorithm_details}.

\begin{figure}
    \centering
\includegraphics[width=0.4\linewidth]{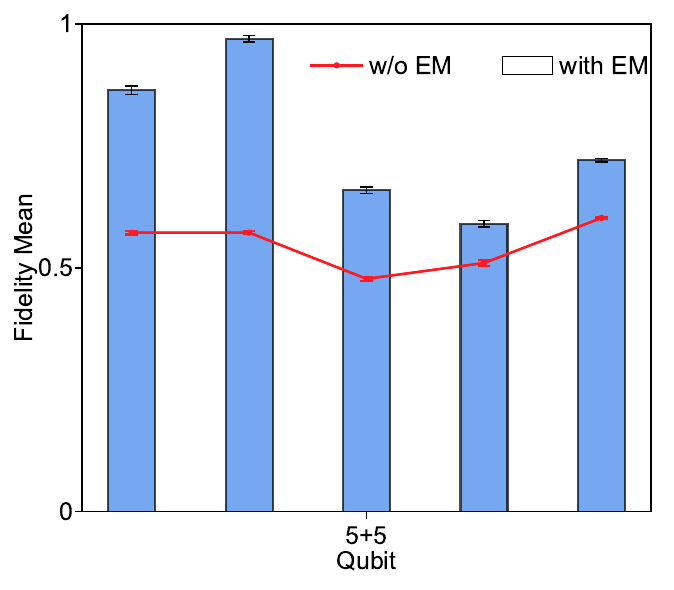}
    \caption{Fidelity with the ideal GHZ state for 7-qubit ($5+5$) GHZ states prepared on different modular quantum processors, shown with and without error mitigation (EM and w/o EM). This figure uses the same experimental data as in Fig. 2 of the main text.}
    \label{suppfig:5plus5_fidpure}
\end{figure}

\begin{table}[htbp]
\caption{Experimental settings used for the estimation error analysis shown in Supplementary Figure~\ref{suppfig:ErrMeas}.}
\begin{tabular}{|c|c|c|c|c|c|}
\hline
\textbf{State} & \textbf{Qubits} & \textbf{$\#\Wcal$ per round}  & \textbf{\#Shots per round} & \textbf{\#Measurements($\times10^5$)} & \textbf{Rounds} \\ \hline
\multirow{5}{*}{GHZ} & \multirow{5}{*}{9 (5+5)} & 10 & \multirow{5}{*}{1000} & 1 & 60 \\ \cline{3-3} \cline{5-6} 
 &  & 20 &   & 2 & 30 \\ \cline{3-3} \cline{5-6} 
 &  & 40 &   & 4 & 12 \\ \cline{3-3} \cline{5-6}  
 &  & 80 &   & 8 & 6 \\ \cline{3-3} \cline{5-6}  
 &  & 100 &   & 10 & 6 \\ \hline
\end{tabular}
\label{tab:fiveplusfive_settings}
\end{table}

\subsubsection{Simulation results in multiple cuttings}

We evaluate the performance of our algorithm in scenarios involving more than one cut. 

Specifically, we first consider the case of two cuts, which partition the 7-qubit GHZ state into three parts. To quantify the estimation error, we repeat the experiment over 30 independent rounds and compute the error as $\varepsilon = \sqrt{\sum_{i=1}^{30} (v_i - 1)^2 / 30}$, where $v_i$ denotes the estimation obtained from the $i$-th trial. In this analysis, we set $\rho = \sigma$ as the ideal GHZ state.

As shown in Supplementary Figure~\ref{suppfig:simulate_multicut}\textbf{a},\textbf{b}, the estimation error of our algorithm decreases with an increasing number of measurement shots and random unitaries per part. Notably, the error reduction is more significant with an increased number of random unitaries, illustrating the effectiveness of the cutting procedure: applying $N$ random unitaries per part yields a total of $N^3$ unique combinations due to the three-part partitioning.

We demonstrate that increasing the number of circuit cuts improves estimation accuracy when using the same number of random unitaries. Supplementary Figure~\ref{suppfig:simulate_multicut}\textbf{c} exhibits the estimation error under a fixed number of random unitaries and measurement shots for cases with one, two, and three cuts of the 9-qubit GHZ state, which corresponds to 2, 3, and 4 parts, respectively. The numerical results support the theoretical results, highlighting the benefit of our approach when the quantum state preparation circuit permits decomposition into multiple parts with only a small number of cuts.

As a generalization of the Method result, here we also utilize the random configuration $\Zcal$ with four different choices, resulting in $32^{r-2} + 12$ number of circuit settings where $r\geq 2$ is the number of parts.  Consequently, for the scenarios depicted in Supplementary Figure~\ref{suppfig:simulate_multicut}\textbf{c}, the number of circuit settings is $12$, $44$, and $76$, respectively.

\begin{figure}
    \centering
    \includegraphics[width=0.6\linewidth]{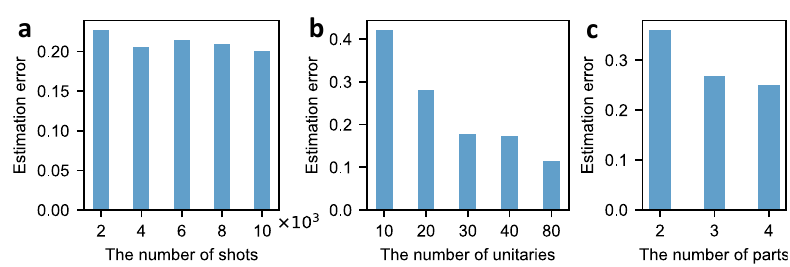}
    \caption{Estimation errors for $\tr{\rho\sigma}$. \textbf{a.} Estimation errors with an increasing number of shots per circuit setting for a 7-qubit GHZ state generated using 2 cuts across three 3-qubit devices, with 30 randomly generated unitaries applied in each segment.
    \textbf{b.} Estimation errors with an increasing number of randomly generated unitaries in each segment, for a 7-qubit GHZ state constructed using 2 cuts across three 3-qubit devices, with $10{,}000$ shots per circuit setting. \textbf{c.} Estimation errors for a $9$-qubit GHZ state constructed using $1$, $2$, and $3$ cuts, corresponding to device configurations of $(5,5)$, $(4,4,3)$, and $(3,3,3,3)$ qubits per part, respectively, with 30 random unitaries applied in each part and $10{,}000$ shots per circuit setting.
    }
    \label{suppfig:simulate_multicut}
\end{figure}

\subsection{Quantum phase learning}

\subsubsection{Problem setup}

Quantum phase recognition plays a central role in understanding complex quantum systems, particularly in condensed matter physics and quantum chemistry. In this work, we focus on learning the phase structure of the transverse-field Ising model, a well-studied family of Hamiltonians with the form
\begin{align} 
H(J, h) = - J \sum_{j=1}^{n-1} Z_j Z_{j+1} - h \sum_{j=1}^{n} X_j, 
\end{align}
where $J$ and $h$ are real, positive parameters controlling the interaction strength and transverse field, respectively. This model exhibits a quantum phase transition at a critical ratio of $J/h$ in the thermodynamic limit.

To investigate the phase behavior, we consider the ground states associated with a specific ansatz circuit and a training dataset $\mathcal{T} = \cbra{\bm x_i, y_i}$, where $\bm x_i$ specifies the state $\ket{\psi\pbra{\bm x^{(i)}}}$, which is prepared by applying a tailored parameterized circuit $U\pbra{\bm x^{(i)}}$, and $y_i$ labels the corresponding quantum phase.

We employ a parameterized variational quantum eigensolver (VQE) circuit to approximate the ground states $\ket{\psi\pbra{\bm x^{(i)}}}$ of the Hamiltonian family $H(h_i):= H(\frac{1}{2}, h_i)$ and use these states to train a model that learns the underlying phase structure.

 \subsubsection{Parameter settings and simulation results}

The hyperparameter settings used in the experiments are as follows. We fix $J=0.5$, and the value of $h$ iterates through [-1.45,-1.35, -1.20, -1.05, -0.90, -0.75, -0.60, -0.45, -0.30, -0.15, 0.00, 0.15, 0.30, 0.45, 0.60, 0.75, 0.90, 1.05, 1.20, 1.35, 1.45].
The classical estimation of the ground state energies and phases is presented in Supplementary Figure~\ref{suppfig:simulate_phaseLearn}. The classically calculated kernel matrix shown in Fig.~3\textbf{c} of the main text is obtained by evaluating the trace overlap between optimized quantum states prepared using VQEs.
The mean squared error (MSE) of two data sets $\cbra{y_i}_{i=1}^R$ and $\cbra{ y'_i}_{i=1}^R$ is defined as
\begin{equation}
    \text{MSE} = \sum_{i=1}^R \pbra{y'_i - y_i}^2.
\end{equation}
For two kernel matrices $K$ and $L$, the centered kernel alignment (CKA) is defined as
\begin{equation} \text{CKA} = \frac{\braket{\tilde{K}, \tilde{L}}_{F}}{\|\tilde{K}\|_F \|\tilde{L}\|_F}, \end{equation}
where $\braket{A, B}_F = \tr{A^{\top} B}$ denotes the Frobenius inner product, and $\|\cdot\|_F$ is the Frobenius norm. The centered kernel matrix $\tilde{K}$ is given by $\tilde{K} = H K H$, where $H = \Ibb_n - \bm{1}_n \bm{1}_n^{\top}$ is the centering matrix, $\bm{1}_n$ is the all-ones column vector in $\Rbb^n$, and $\Ibb_n$ is the $n \times n$ identity matrix. A similar expression applies for $\tilde{L}$.

\subsubsection{Experimental circuit settings in the main file}

\begin{figure}
    \centering
\includegraphics[width=1.0\linewidth]{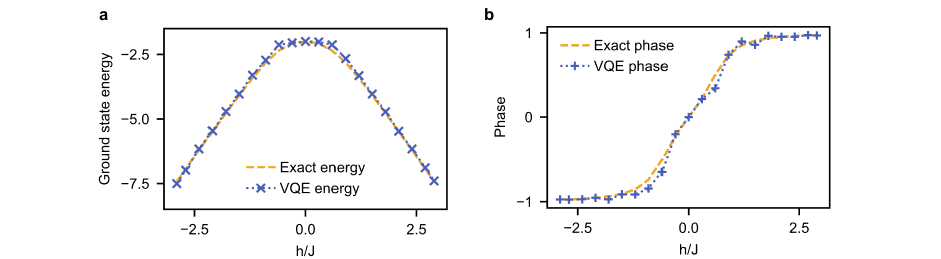}
    \caption{Simulation results for the ground state energy and phases with the VQE algorithm. \textbf{a.} Ground state energies with the VQE algorithm. \textbf{b.} Simulated phases with the VQE algorithm.}
    \label{suppfig:simulate_phaseLearn}
\end{figure}
 
The experimental circuits designed for implementing quantum machine learning algorithms for phase learning through classical communication are illustrated in Supplementary Figure~\ref{suppfig:SI_phase_learning_circuit}. These circuits consist of multiple layers of single- and two-qubit gates.
The single-qubit gates include rotations about $\sigma_z$ for angles determined by the designed ansatz, as well as rotations about $\sigma_x$ and $\sigma_y$ for specified angles. The $\sigma_x$ and $\sigma_y$ rotations are achieved by applying 30 ns microwave pulses at qubits' idle frequencies, employing the ``derivative reduction by
adiabatic gate technique'' to mitigate leakage and enhance fidelity~\cite{song201710}.
$\sigma_z$ rotations are implemented using virtual Z gates~\cite{mckay2017efficient}.
The two-qubit CZ gate is realized by bringing a pair of neighboring qubits into resonance at the $|11\rangle$ and $|02\rangle$($|20\rangle$) states, and applying a 50 ns square pulse to the betweening coupler~\cite{Arute2019Quantum,foxen2020demonstrating}.

As shown in Supplementary Figure~\ref{suppfig:SI_phase_learning_circuit}, for the phase learning circuit, we employ a classical computer to assist in generating the encoded parameters, which are subsequently mapped to the angles of the $R_x^\theta$ and $R_z^\theta$ gates.

\begin{figure}[htpb]
    \centering
    \includegraphics[width=0.8\linewidth]{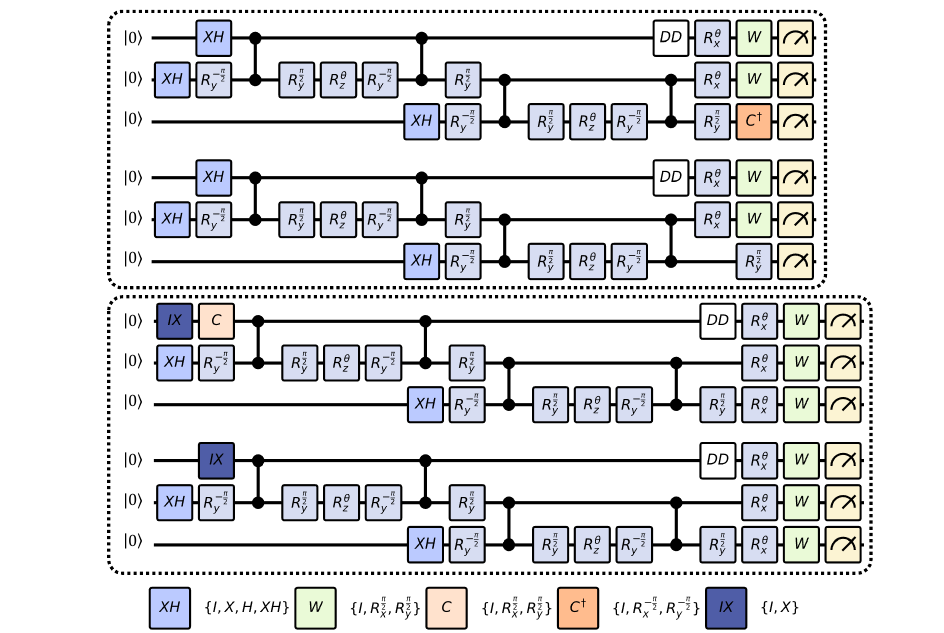}
    \caption{Experimental circuit for phase learning in the main text. The first dotted box shows the \textsf{Q}-\textsf{A} circuit after cutting, where $C^\dagger \in \{\mathbb{I}, R_x^{-\pi/2}, R_y^{-\pi/2}\}$, yielding 4 possible circuit base configurations. The second dotted box shows the \textsf{Q}-\textsf{B} circuit after cutting, with choices of $C$ and input states of the cutting qubit set to $0$ or $1$, resulting in 8 configurations. In total, there are 12 distinct circuit base configurations. The ansatz $R_z(\theta)$ is varied to optimize toward the ground state.
} \label{suppfig:SI_phase_learning_circuit}
\end{figure}

\end{document}